\documentclass[11pt]{article}

\usepackage{fullpage}

\usepackage{amsthm,amsfonts,amssymb,amsmath,amsxtra}
\usepackage{varwidth}
\usepackage[dvipsnames]{xcolor}
\usepackage{pgfplots}
\usepackage{xspace,enumerate}
\usepackage[utf8]{inputenc}
\usepackage{thmtools}
\usepackage{thm-restate}
\usepackage{authblk}
\definecolor{TUMBlau}{HTML}{0065bd}
\usepackage[hypertexnames=false,colorlinks=true,urlcolor=TUMBlau,citecolor=TUMBlau,linkcolor=TUMBlau]{hyperref}
\usepackage[noadjust]{cite}
\usepackage{microtype}
\usepackage{todonotes}
\usepackage{thmtools}
\usepackage[capitalise]{cleveref}
\usepackage{algorithm}
\usepackage[noend]{algpseudocode}
\usepackage[pagewise]{lineno}
\usetikzlibrary{patterns}
\nolinenumbers

\theoremstyle{plain}
\newtheorem{theorem}{Theorem}
\newtheorem{lemma}[theorem]{Lemma} 
\newtheorem{corollary}[theorem]{Corollary}  
\newtheorem{proposition}[theorem]{Proposition}  
\newtheorem{definition}[theorem]{Definition}
\newtheorem{mainlemma}{Main Lemma}

\newtheorem{remark}{Remark}

\newcommand{\bE}{\ensuremath{\mathbf{E}}\xspace}
\newcommand{\bP}{\ensuremath{\mathbf{P}}\xspace}
\newcommand{\CJ}{\ensuremath{\mathcal{J}}\xspace}
\newcommand{\OPT}{{\mathrm{OPT}}}
\newcommand{\rom}{{\mathrm{rom}}}

\newcommand{\BZ}{\ensuremath{\mathbb{Z}}\xspace}
\newcommand{\CC}{\ensuremath{\mathcal{C}}\xspace}

\newcommand{\CM}{\ensuremath{\mathcal{M}}\xspace}
\newcommand{\CP}{\ensuremath{\mathcal{P}}\xspace}

\newcommand{\jsmall}{{\mathrm{small}}}
\newcommand{\jmed}{{\mathrm{med}}}
\newcommand{\jbig}{{\mathrm{big}}}
\newcommand{\jmax}{{\mathrm{max}}}
\newcommand{\lowM}{{\mathrm{low}}}
\newcommand{\midM}{{\mathrm{mid}}}
\newcommand{\mempty}{{\mathrm{empty}}}
\newcommand{\guess}{{\mathrm{guess}}}
\newcommand{\LL}[1]{{\mathrm{Light Load[}#1\mathrm{]}}}
\newcommand{\pre}{{\mathrm{pre}}}
\newcommand{\res}{{\mathrm{res}}}
\newcommand{\rep}{{\mathrm{rep}}}

\author[1]{Susanne Albers}
\author[2]{Maximilian Janke}
\affil[1]{Department of Computer Science, Technical University of Munich\\
\href{mailto:albers@in.tum.de}{albers@in.tum.de}}

\affil[2]{Department of Computer Science, Technical University of Munich\\
\href{mailto:janke@in.tum.de}{janke@in.tum.de}}

\date{}

\title{Scheduling in the Secretary Model\footnote{Work supported by the European Research Council, Grant Agreement No. 691672, project APEG.}}

\begin{document}

\maketitle
\begin{abstract}
This paper studies Makespan Minimization in the random-order model. Formally, jobs, specified by their processing times, are presented in a uniformly random order. An online algorithm has to assign each job permanently and irrevocably to one of~$m$ parallel and identical machines such that the expected time it takes to process them all, the makespan, is minimized.

We give two deterministic algorithms. First, a straightforward adaptation of the semi-online strategy $\mathrm{Light Load}$ \cite{albers_semi-online_2012} provides a very simple algorithm retaining its competitive ratio of $1.75$. A new and sophisticated algorithm is $1.535$-competitive. These competitive ratios are not only obtained in expectation but, in fact, for all but a very tiny fraction of job orders.

Classically, online makespan minimization only considers the worst-case order. Here, no competitive ratio below $1.885$ for deterministic algorithms and $1.581$ using randomization is possible. The best randomized algorithm so far is $1.916$-competitive. Our results show that classical worst-case orders are quite rare and pessimistic for many applications. They also demonstrate the power of randomization when compared much stronger deterministic reordering models.

We complement our results by providing first lower bounds. A competitive ratio obtained on nearly all possible job orders must be at least $1.257$. This implies a lower bound of~$1.043$ for both deterministic and randomized algorithms in the general model.
\end{abstract}

\section{Introduction}
We study one of the most basic scheduling problems, the classic problem of makespan minimization. For the classic makespan minimization problem one is given an input set $\CJ$ of $n$ jobs, which have to be scheduled onto $m$ identical and parallel machines. Preemption is not allowed. Each job $J\in\CJ$ runs on precisely one machine. 
The goal is to find a schedule minimizing the \emph{makespan}, i.e.\ the last completion time of a job. This problem admits a long line of research and countless practical applications in both, its offline variant see e.g.\ \cite{graham_bounds_1966, hochbaum_using_1987} and references therein, as well as in the online setting studied in this paper.

In the online setting jobs are revealed one by one and each has to be scheduled by an online algorithm $A$ immediately and irrevocably without knowing the sizes of future jobs. The makespan of online algorithm $A$, denoted by $A(\CJ^\sigma)$, may depend on both the job set $\CJ$ and the job order~$\sigma$. The optimum makespan $\OPT(\CJ)$ only depends on the former. Traditionally, one measures the performance of $A$ in terms of competitive analysis. The input set $\CJ$ as well as the job order~$\sigma$ are chosen by an adversary whose goal is to maximize the ratio $\frac{A(\CJ^\sigma)}{\OPT(\CJ)}$. The maximum ratio, $c=\sup_{\CJ,\sigma} \frac{A(\CJ^\sigma)}{\OPT(\CJ)}$, is the \emph{(adversarial) competitive ratio}. The goal is to find online algorithms obtaining small competitive ratios. 

In the classical secretary problem the goal is to hire the best secretary out of a linearly ordered set $S$ of candidates. Its size $n$ is known. Secretaries appear one by one in a uniformly random order. An online algorithm can only compare secretaries it has seen so far. It has to decide irrevocably for each new arrival whether this is the single one it wants to hire. Once a candidate is hired, future ones are automatically rejected even if they are better. The algorithm fails unless it picks the best secretary. Similar to makespan minimization this problem has been long studied, see \cite{ dynkin_optimum_1963, feldman_simple_2014, ferguson_who_1989, kaplan_competitive_2020, kleinberg_multiple-choice_2005, lachish_o_2014, lindley_dynamic_1961} and references therein.

This paper studies a makespan minimization under the input model of the secretary problem. The adversary determines a job set of known size $n$.  Similar to the secretary problem, these jobs are presented to an online algorithm $A$ one by one in a uniformly random order. Again, $A$ has to schedule each job without knowledge of the future. The expected makespan is considered. The \emph{competitive ratio in the secretary (or random-order) model}  $c=\sup_{\CJ} \bE_\sigma\left[ \frac{A(\CJ^\sigma)}{\OPT(\CJ)}\right]$ is the maximum ratio between the expected makespan of~$A$ and the optimum makespan. The goal is again to obtain small competitive ratios.

We propose the term \emph{secretary model}, first used in \cite{molinaro_online_2017}, to set this model apart from the model studied by the same authors in~\cite{albers_scheduling_2020}  where~$n$, the number of jobs, is not known in advance. Not knowing $n$ is quite restrictive and has never been considered in any other work on scheduling with random-order arrival~\cite{gobel2015online, molinaro_online_2017, osborn_lists_2008}. For the adversarial model such information is useless.

Similar frameworks received a lot of recent attention in the research community sparking the area of random-order analysis. Random-order analysis has been successfully applied to numerous problems such as matching \cite{goel_online_2008, karande_online_2011, karp_optimal_1990, mahdian_online_2011}, various generalizations of the secretary problem \cite{babaioff_matroid_2018, feldman_simple_2014, ferguson_who_1989, kaplan_competitive_2020, kleinberg_multiple-choice_2005, lachish_o_2014}, knapsack problems \cite{babaioff_knapsack_2007}, bin packing \cite{kenyon_best-fit_1996}, facility location \cite{meyerson_online_2001},  packing LPs \cite{kesselheim_primal_2014}, convex optimization \cite{gupta_maximizing_2018}, welfare maximization \cite{korula_online_2018}, budgeted allocation \cite{mirrokni_simultaneous_2012} and recently scheduling \cite{albers_scheduling_2020, gobel2015online, molinaro_online_2017, osborn_lists_2008}. 

For makespan minimization the role of randomization is poorly understood.  The lower bound of $1.581$ from~\cite{chen_lower_1994,sgall_lower_1997} is considered pessimistic and exhibits quite a big gap towards the best randomized ratio of $1.916$ from~\cite{albers_randomized_2002}.

The upper bound of $1.535$ in this paper demonstrates surprising power when it comes to randomization in the input order.  The power of reordering has been studied by Englert~et.~al.~\cite{englert_power_2008}. Their lower bound considers online algorithms, which are able to look-ahead and rearrange almost all of the input sequence in advance. Their only disadvantage is that such rearrangement is deterministic. Englert~et.~al.\  show that these algorithms can not be better than $1.466$-competitive for general~$m$. This is quite close to our upper bound of $1.535$, given that the algorithm involved has neither look-ahead nor control over the arrangement of the sequence. 

A main consequence of the paper is that random-order arrival allows to beat the lower bound of $1.581$ for randomized adversarial algorithms. This formally sets this model apart from the classical adversarial setting even if randomization is involved.

\subparagraph*{Previous work:}
Online makespan minimization and variants of the secretary problem have been studied extensively. We only review results most relevant to this work beginning with the traditional adversarial setting. For $m$ identical machines, Graham \cite{graham_bounds_1966} showed 1966 that the greedy strategy, which schedules each job onto a least loaded machine, is $\left(2-\frac{1}{m}\right)$-competitive.  This was subsequently improved in a long line of research \cite{galambos_-line_1993, bartal_new_1992, karger_better_1996, albers_better_1999} leading to the currently best competitive ratio by Fleischer and Wahl \cite{fleischer_-line_2000}, which approaches $1.9201$ for $m\rightarrow\infty$. Chen et al.\ \cite{chen_approximating_2015} presented an algorithm whose competitive ratio is at most $(1+\varepsilon)$-times the optimum one, though the actual ratio remains to be determined. For general~$m$, lower bounds are provided in \cite{faigle_performance_1989, bartal_better_1994, gormley_generating_2000, rudin_improved_2001}. The currently best bound is due to Rudin III \cite{rudin_improved_2001} who shows that no deterministic online algorithm can be better than $1.88$-competitive.

The role of randomization in this model is not well understood. 
The currently best randomized ratio of $1.916$ \cite{albers_randomized_2002} barely beats deterministic guarantees. In contrast, the best lower bound approaches $\frac{e}{e-1}> 1.581$ for $m\rightarrow\infty$ \cite{chen_lower_1994,sgall_lower_1997}. There has been considerable research interest in tightening the~gap.
 
Recent results for makespan minimization consider variants where the online algorithm obtains extra resources. There is the semi-online setting where additional information on the job sequence is given in advance, like the optimum makespan \cite{kellerer_efficient_2013} or the total processing time of jobs \cite{albers_semi-online_2012, cheng_semi--line_2005, kellerer_semi_1997,kellerer2015efficient}. In the former model the optimum competitive ratio lies in the interval $[1.333,1.5]$, see~\cite{kellerer_efficient_2013}, while for the latter the optimum competitive ratio is known to be~$1.585$~\cite{albers_semi-online_2012,kellerer2015efficient}.
Taking this further, the advice complexity setting allows the algorithm to receive a certain number of advice bits from an offline oracle \cite{albers_online_2017, dohrau_online_2015, kellerer_semi_1997}.  

The work of Englert~et.~al.~\cite{englert_power_2008} is particularly relevant as they, too, study the power of reordering. Their algorithm has a buffer, which can reorder the sequence 'on the fly'. They prove that a buffer size linear in $m$ suffices to be $1.466$-competitive. Their lower bound shows that this result cannot be improved for any sensible buffer size.\footnote{A buffer size of $n$ would not be sensible since it reverts to the offline problem, which admits a PTAS~\cite{hochbaum_using_1987}. Their lower bound holds for any buffer size $b(n)$, depending on the input size $n$, if $n-b(n)$ is unbounded. Such a buffer can already hold almost all, say any fraction, of the input sequence.}

The secretary problem is even older than scheduling \cite{ferguson_who_1989}. Since the literature is vast, we only summarize the work most relevant to this paper. Lindley \cite{lindley_dynamic_1961} and Dynkin \cite{dynkin_optimum_1963} first show that the optimum strategy finds the best secretary with probability $1/e$ for $n\rightarrow\infty$.  Recent research focusses on many variants, among others generalizations to several secretaries \cite{albers2020new, kleinberg_multiple-choice_2005} or even matroids~\cite{babaioff_matroid_2018, feldman_simple_2014, lachish_o_2014}. A modern version considers adversarial orders but allows prior sampling \cite{correa2021secretary, kaplan_competitive_2020}. Related models are prophet inequalities and the game of googol \cite{correa2020two,correa2019prophet}.

So far, little is known for scheduling in the secretary model. Osborn and Torng \cite{osborn_lists_2008} prove that Graham's greedy strategy is still not better than $2$-competitive for~$m\rightarrow\infty$. In \cite{albers_scheduling_2020} we studied the very restricted variant where~$n$, the number of jobs, is not known in advance and provide a $1.8478$-competitive algorithm and first lower bounds. Here, most common techniques, e.g.\ sampling, do not work. We are the only ones who ever considered this restriction. Molinary \cite{molinaro_online_2017} studies a very general scheduling problem. His algorithm has expected makespan $(1+\varepsilon)\OPT+O(\log(m)/\varepsilon)$, but its random-order competitive ratio is not further analyzed. G\"obel~et~al.~\cite{gobel2015online} study a scheduling problem on a single machine where the goal is to minimize weighted completion times. Their competitive ratio is $O(\log(n))$ whereas they show that the adversarial model allows no sublinear competitive ratio.

\subparagraph*{Our contribution:}
We study makespan minimization for the secretary (or random-order) model in depth. We show that basic sampling ideas allow to adapt a fairly simple algorithm from the literature \cite{albers_semi-online_2012} to be $1.75$-competitive. A more sophisticated algorithm vastly improves this competitive ratio to $1.535$. This beats all lower bounds for adversarial scheduling, including the bound of $1.582$ for randomized algorithms.

Our main results focus on a large number of machines, $m\rightarrow\infty$. This is in line with most recent adversarial results \cite{albers_randomized_2002, fleischer_-line_2000} and all random-order scheduling results \cite{albers_scheduling_2020, gobel2015online, molinaro_online_2017, osborn_lists_2008}. While adversarial guarantees are known to improve for small numbers of machines, nobody has ever, to the best of our knowledge, explored guarantees for random-order arrival on a small number of machines. We prove that our simple algorithm is $\left(1.75+O\left(\frac{1}{\sqrt{m}}\right)\right)$-competitive. Explicit bounds on the hidden term are given as well as simulations, which indicate good performance in practice. This shows that the focus of contemporary analyses on the limit case is sensible and does not hide unreasonably large additional terms.

All results in this paper abide to the stronger measure of \emph{nearly competitiveness} from~\cite{albers_scheduling_2020}. An algorithm is required to achieve its competitive ratio not only in expectation but on nearly all input permutations. Thus, input sequences where it is not obtained can be considered extremely rare and pathological. Moreover, we require worst-case guarantees even for such pathological inputs. This seems quite relevant to practical applications, where we do not expect fully random inputs. Both algorithms in this paper hold up to this stronger measure of nearly competitiveness. 

A basic approch in random-order models relies on sampling; a small part of the input is used to predict the rest. Sampling allows us to include techniques from semi-online and advice scheduling with two further challenges. On the one hand, the advice is imperfect and may be, albeit with low probability, totally wrong. On the other hand, the advice has to be learned, rather than being available right from the start. In the beginning 'mistakes' cannot be avoided. This makes it impossible to adapt better semi-online algorithms than $\mathrm{Light Load}$, namely~\cite{cheng_semi--line_2005, kellerer_semi_1997,kellerer2015efficient} to our model. These algorithms need to know the total processing volume right from the start.
Note that the advanced algorithm in this paper out-competes the optimum competitive ratio of~$1.585$ these semi-online algorithms can achieve~\cite{albers_better_1999,kellerer2015efficient}. We conjecture that this is not possible for algorithms that solely use sampling.
 
Algorithms that can only use sampling are studied in a modern variant of the secretary problem \cite{correa2021secretary, kaplan_competitive_2020}. First, a random sample is observed, then the sequence is treated in adversarial order. The analysis of $\mathrm{Light Load}$ carries over to such a model without changes. The $1.535$-competitive algorithm does not maintain its competitive ratio in such a model. 

The $1.535$-competitive main algorithm is based on a modern point of view, which, analogous to kernelization, reduces complex inputs to sets of critical jobs. A set of critical jobs is estimated using sampling. Critical jobs impose a lower bound on the optimum makespan. If the bound is high, an enhanced version of Graham's greedy strategy suffices; called the Least-Loaded-Strategy. Else, it is important to schedule critical jobs correctly. The Critical-Job-Strategy, based on sampling, estimates the critical jobs and schedules them ahead of time. An easy heuristic suffices, due to uncertainty involved in the estimates. Uncertainty poses not only the main challenge in the design of the Critical-Job-Strategy. On a larger scale, it also makes it hard to decide, which of the two strategies to use. Sometimes the Critical-Job-Strategy is chosen wrongly. These cases comprise the crux of the analysis and require using random-order arrival in a novel way beyond sampling.

The analyses of both algorithms follows three steps. First, adversarial analyses give worst-case guarantees and take care of \emph{simple job sets}, which lack structure to be exploited via random reordering. Intuitively, random sequences have certain properties, like being not 'ordered'. A second step formalizes this, introducing stable orders. Non-stable orders are rare and negligible. Reducing to stable orders yields a natural semi-online setting. Third, we analyze our algorithm in this semi-online setting. See \Cref{fig:simplestable} for a lay of the land.
 
The paper concludes with lower bounds for the secretary~model. No algorithm, deterministic or randomized, is better than nearly $1.257$-competitive. This immediately yields a lower bound of $1.043$ in the general secretary model, too.

\section{Notation}
Almost all notations relevant in scheduling depend on the input set $\CJ$ or on the ordered input sequence $\CJ^\sigma$. We use the notation $[\CJ]$ and $[\CJ^\sigma]$ to indicate such dependency, for example $L[\CJ]$ and $L_\varphi[\CJ^\sigma]$. If such dependency needs not be mentioned, for example if the sequence $\CJ^\sigma$ is fixed, we drop this appendage, simply writing $L$ and $L_\varphi$. Similarly, we write $\OPT$ for $\OPT(\CJ)$. If we focus on the job order $\sigma$ whilst the dependency of the job set $\CJ$ does not deserve mention, the notation $[\sigma]$ instead of $[\CJ^\sigma]$ is used. We could for example write $L_\varphi[\sigma]$.

\section{A strong measure of random-order competitiveness}

Consider a set of $n$ jobs $\CJ=\{J_1,\ldots, J_n\}$ with non-negative sizes $p_1,\ldots, p_n$ and let $S_n$ be the group of permutations of the integers from $1$ to $n$. We consider $S_n$ a probability space under the uniform distribution, i.e.\ we pick each permutation with probability $1/n!$. Each permutation $\sigma\in S_n$, called an \emph{order}, gives a \emph{job sequence} $\CJ^\sigma=J_{\sigma(1)},\ldots,J_{\sigma(n)}$. Recall that traditionally an online algorithm $A$ is called \emph{$c$-competitive} for some $c\ge 1$ if we have for all job sets $\CJ$ and job orders $\sigma$ that $A(\CJ^\sigma)\le c\OPT(\CJ)$. We call this the \emph{adversarial model}.

In the secretary~model we consider the expected makespan of $A$ under a uniformly chosen job order, i.e.\ $A^\rom=\bE_{\sigma\sim S_n}[A(\CJ^\sigma)]=\frac{1}{n!}\sum_{\sigma\in S_n} A(\CJ^\sigma)$, rather than the makespan achieved in a worst-case order. The algorithm $A$ is \emph{$c$-competitive in the secretary model} if $A^\rom(\CJ)\le c\OPT(\CJ)$ for all input sets $\CJ$.

This model tries to lower the impact of particularly badly ordered sequences by looking at competitive ratios only in expectation. Interestingly, the scheduling problem allows for a stronger measure of random-order competitiveness for large $m$, called \emph{nearly competitiveness}~\cite{albers_scheduling_2020}. One requires the given competitive ratio to be obtained on nearly all sequences, not only in expectation, as well as a bound on the adversarial competitive ratio as well. We recall the definition and the main fact, that an algorithm is already $c$-competitive in the secretary model if it is nearly $c$-competitive.

\begin{definition}\label{def:comp}
A deterministic online algorithm $A$ is called {\em nearly $c$-competitive\/} if the following two conditions hold.
\begin{itemize}
\setlength{\parskip}{0pt} \setlength{\itemsep}{0pt plus 1pt}
\item The algorithm $A$ achieves a constant competitive ratio in the adversarial model.
\item For every $\varepsilon >0$, we can find $m(\varepsilon)$ such that for all machine numbers $m \geq m(\varepsilon)$
and all job sequences ${\cal J}$ there holds $\bP_{\sigma\sim S_n} [A({\cal J}^\sigma) \geq (c+\varepsilon) OPT({\cal J})] \leq \varepsilon$.
\end{itemize}
\end{definition}

\begin{lemma}\label{le.nearly}
If a deterministic online algorithm is nearly $c$-competitive, then it is $c$-competitive in the random-order model as $m\rightarrow \infty$.
\end{lemma}

\begin{proof}
Let $C$ be the constant adversarial competitive ratio of $A$. Given $\delta>0$ we need to show that we can choose $m$ large enough such that our algorithm is $(c+\delta)$-competitive in the random-order model. For $\varepsilon=\frac{\delta}{C-c+1}$ choose $m$ large enough such that $P_\varepsilon(\CJ)=\bP_{\sigma\sim S_n}\left[C_A(\CJ^\sigma)\ge (c+\varepsilon)\OPT(\CJ)\right] \le \varepsilon$ holds for every input sequence $\CJ$. Then we have for every input sequence $\CJ$ that
\begin{align*}C^\rom(\CJ)&\le \left(1-P_\varepsilon(\CJ)\right)\cdot (c+\varepsilon)\OPT(\CJ) +P_\varepsilon(\CJ)\cdot C\cdot\OPT(\CJ)\le ((1-\varepsilon)(c+\varepsilon)+\varepsilon C)\OPT(\CJ) \\ &\le(c+\delta(C-c+1))\OPT(\CJ)=(c+\delta)\OPT(\CJ). &\qedhere \end{align*}
\end{proof}

\section{Basic properties}\label{sec.basic}
Given an input sequence $\CJ^\sigma=J_{\sigma(1)},\ldots J_{\sigma(n)}$ and $0<\varphi\le 1$, we consider the \emph{load estimate} $L_\varphi=L_\varphi[\CJ^\sigma]=\frac{1}{\varphi m}\sum_{\sigma(t)\le\varphi n} p_{t}$, which is $\varphi^{-1}$-times the average load (in any schedule) after the first $\varphi n$ jobs have been assigned. We are particularly interested in the \emph{average load} $L=L[\CJ]=L_1[\CJ^\sigma]$, which is a lower bound for $\OPT$. The value $L_\varphi$ for smaller $\varphi$ is a guess for $L$, which can be made by an online algorithm after a $\varphi$-fraction of the input-sequence has been observed.
Given $t> 0$, let $p^t_\jmax=\max(p_{t'}\mid t'< t+1)$ be the size of the largest among the first $\lceil t\rceil$ jobs. In particular, $p_\jmax=p^n_\jmax$, the \emph{size of the largest jobs}, is again an important lower bound for $\OPT$.

\begin{proposition}\label{pro.bounds}
We have the following lower bounds for the optimum makespan: \vspace{-5pt}
\begin{itemize}
\item $p_\jmax \le \OPT$\vspace{-5pt}
\item $L \le \OPT$
\end{itemize}
\end{proposition}

\begin{proof}The first bound follows from observing that any schedule must, in particular, schedule the largest job on some machine whose load thus is at least $p_\jmax$. For the second bound one observes that the makespan, the maximum load of a machine in any given schedule, cannot be less than $L$, the average load of all machines.
\end{proof}

Let us consider any fixed (ordered) job sequence $\CJ^\sigma=J_{\sigma(1)},\ldots J_{\sigma(n)}$ and any (deterministic) algorithm that assigns these jobs to machines. We begin with some fundamental observations.

\begin{lemma}\label{le.avglb}
Let $\varphi>0$ and $t\le \varphi n$. Then the $k$-th least loaded machine at time $t$ has load at most $\frac{m}{m-k+1}\varphi L_{\varphi}$. In particular, its load is at most $\frac{m}{m-k+1} L$.
\end{lemma}

\begin{proof}
Let $L^t$ be the sum of all loads at time $t$. Since this is the same as the sum of all processing times of jobs arriving at time $t$, we have $L^t\le\sum_{\sigma(t')\le t}p_{t'} \le \varphi m \frac{1}{\varphi m}\sum\limits_{\sigma(t')\le\varphi n} p_{t'}= \varphi m L_\varphi$.
Let $l$ be the load of the $k$-th least loaded machine at time $t$. Per definition $m-k+1$ machines had at least that load. Thus $(m-k+1)l\le L^t \le  \varphi m L_\varphi$ or, equivalently, $l\le\frac{m}{m-k+1}\varphi L_{\varphi}$.
\end{proof}

Consider the value $R ({\cal J}) =\min\{{L\over p_{\max}},1\}$, which measures the complexity of the input set independent of its order. Informally, a smaller value $R(\CJ)$ makes the job set easier to be scheduled but less suited to reordering arguments. Later, sets with a small value $R(\CJ)$ need to be treated separately. The following proposition is both interesting for its implication on general sequences and, particularly, simple sequences with $R ({\cal J})$ small. 

\begin{figure}[b]
    \centering
    \resizebox{\textwidth}{!}{
\begin{tikzpicture}
  \foreach \x in {0,...,161} {
  \ifnum\x=82
\draw[fill=lightgray!20] (\x,0) rectangle (\x+1,10);
\else
\draw[fill=lightgray!100] (\x,0) rectangle (\x+1,0.5);
\fi
}
\end{tikzpicture}
    }
    \caption{A surprisingly difficult sequence for random-order arguments. The big job carries most of the processing volume. Other jobs are negligible. Thus, all permutations look basically the same. Note that for such a 'simple' sequence $R ({\cal J}) $ is small.
    }
    \label{fig:counterx}
    \end{figure}
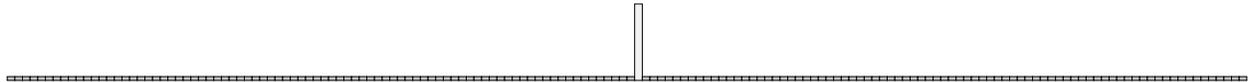

\begin{proposition}\label{prop1}
If any job $J$ is scheduled on the $k$-th least loaded machine, the load of said machine does not exceed $\left(\left(\frac{m}{m-i+1}\right)R(\CJ)+1 \right)OPT({\cal J})$ afterwards.
\end{proposition}

\begin{proof}
Let $l$ be the load of the $i$-th least loaded machine before $J$ is scheduled. Then $l \le \frac{m}{m-k+1}L$ by \Cref{le.avglb}. Since $J$ had size at most $p_\mathrm{max}$, the load of the machine it was scheduled on won't exceed
$l+p_\mathrm{max} \le \frac{m}{m-k+1}L + p_\mathrm{max}\le \frac{m}{m-k+1}\frac{L}{\max(L,p_\mathrm{max})}\OPT+\OPT =  \left(\left(\frac{m}{m-k+1}\right)R(\CJ)+1 \right)OPT$.
\end{proof}

Thus, if an algorithm avoids a constant fraction of most loaded machines, its competitive ratio is bounded and approaches $1$ as~$R ({\cal J})\rightarrow 0$.

We call a vector $(\tilde l_M^t)$ indexed over all machines~$M$ and all times $t=0,\ldots, n$ a \emph{pseudo-load} if $\tilde l_M^t\ge l_M^t$ for any time $t$ and machine $M$. We introduce such a pseudo-load in the analysis of our main algorithm. Let $\tilde L = \sup_t \frac{1}{m} \sum_M \tilde l_M^t$ be the \emph{maximum average pseudo-load} and, again, consider $\tilde R ({\cal J}) = \min\{ {\tilde L\over p_{\max}},{\tilde L\over L}\}=R(\CJ) {\tilde L\over L} $. The following observation is immediate.
\begin{lemma}\label{le.sensibletildeL}
We have $\tilde L\ge L$ and $R({\cal J}) \le \tilde R({\cal J})$.
\end{lemma}
\begin{proof}
Indeed, $L=\frac{1}{m}\sum_M l^n_M\le \frac{1}{m}\sum_M \tilde l^n_M \le \tilde L$. This already implies $R({\cal J}) \le \tilde R({\cal J})$.
\end{proof}

It will be important to note that \Cref{le.avglb} and Proposition~\ref{prop1} generalize to pseudo-loads. Since the proofs stay almost the same, we do not include them in the main body of the paper but leave them to \Cref{sec.basic.p} for completeness.

\begin{restatable}{lemma}{leavglbII}\label{le.avglb2}
Let $\varphi>0$ and $t\le \varphi n$. Then the machine with the $k$-th least pseudo-load at time $t$ had pseudo-load at most $\frac{m}{m-k+1}\tilde L$.
\end{restatable}

\begin{restatable}{proposition}{propII}\label{prop2}
If job $J_{\sigma(t+1)}$ is scheduled on the machine M with $i$-th smallest pseudo-load~$\tilde l_M^t$ at time $t$, then, afterwards, its load $l_M^{t+1}$ does not exceed $\left(1+\left(\frac{m}{m+1-i}\right)\tilde R(\CJ)\right)OPT({\cal J})$.
\end{restatable}

\subsection{Sampling and the Load Lemma}\label{sec.sampling}
Our model is particularly suited to sampling. Given a job set $\CJ$, we call a subset $\CC\subseteq \CJ$ a \emph{job class}. Consider any job order $\sigma\in S_n$. For $0<\varphi\le 1$, let $n_{\CC,\varphi}[\sigma]$ denote the number of jobs in $\CC$ arriving till time $\varphi n$, i.e.\ $n_{\CC,\varphi}[\sigma]=|\{J_{\sigma(i)}\mid J_{\sigma(i)}\in \CC \land\sigma(i) \le \varphi n\}|$. Let $n_{\CC}=n_{\CC,1}[\sigma]=|\CC|$ be the total number of jobs in $\CC$. The following is a consequence of Chebyshev's inequality. The proof is left to \Cref{sec.basic.p}.

\begin{restatable}{proposition}{prosample}\label{pro.sample}
Let $\CC\subset\CJ$ be a job class for a job set $\CJ$ of cardinality at least $m$. Given $\varphi>0$ and $E\ge 0$ we have
\[\bP_{\sigma\sim S_n}\left[\left|\varphi^{-1} n_{\CC,\varphi}[\sigma]-n_{\CC}\right|\ge E\right]\le \frac{n_\CC}{\varphi(E-1/m)^2}.\]
\end{restatable}

A basic lemma in random-order scheduling is the Load Lemma from~\cite{albers_scheduling_2020}, which allows a good estimate of the average load under very mild assumptions on the job set. Here, we introduce a more general version. It is all we need to adapt the semi-online algorithm $\mathrm{LightLoad}$ from the literature to the secretary model.

\begin{figure}[t]
\includegraphics[width=0.3\textwidth]{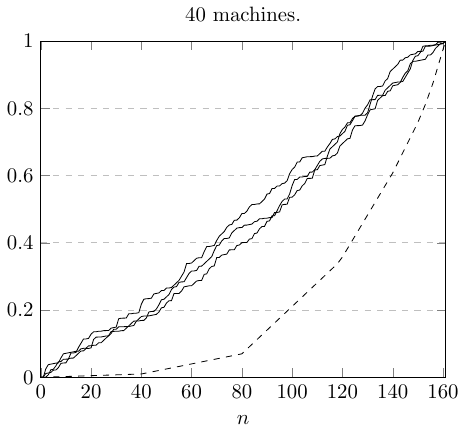}
\includegraphics[width=0.3\textwidth]{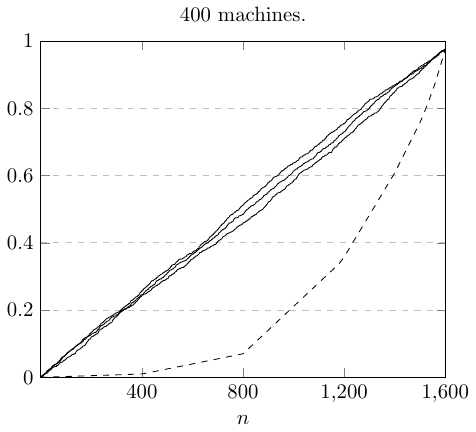}
\includegraphics[width=0.3\textwidth]{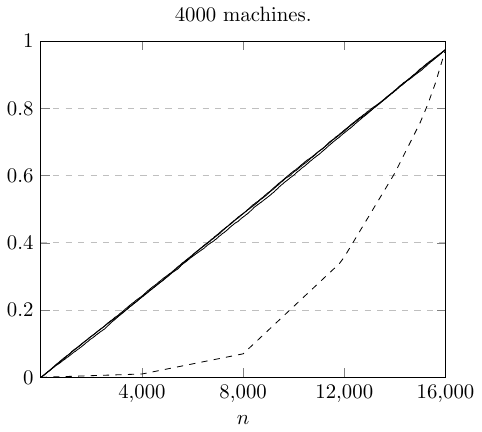}
\caption{A graphic depicting the average load over time on the classical lower bound sequence from \cite{albers_better_1999} for $40$, $400$ and $4000$ machines. The dashed line corresponds to the original adversarial order. The three solid lines corresponding to random permutations clearly approximate a straight line. Thus, sampling allows to predict the (final) average load.
 }\label{fig.ratios}
\end{figure}

\begin{restatable}{lemma}{Loadlemma}[Load Lemma~\cite{albers_scheduling_2020}]\label{Loadlemma}
Let $R_\mathrm{low}=R_\mathrm{low}(m)>0$, $1\ge \varphi=\varphi(m)>0$ and $\varepsilon=\varepsilon(m)>0$ be three functions such that $\varepsilon^{-4}\varphi^{-1}R_\mathrm{low}^{-1}=o(m)$. Then there exists a variable $m(R_\mathrm{low},\varphi,\varepsilon)$ such that we have for all $m\ge m(R_\mathrm{low},\varphi,\varepsilon)$ and all job sets $\CJ$ with $R(\CJ)\ge R_\mathrm{low}$ and $|\CJ|\ge m$:
\[ \bP_{\sigma\sim S_n} \left[\left|\frac{L_\varphi[\CJ^\sigma]}{L[\CJ]}-1\right|\ge\varepsilon \right]<\varepsilon.\]
\end{restatable}

We sketch the proof, leaving the details to \Cref{sec.basic.p} since it is technical and a slight generalization of the one found in~\cite{albers_scheduling_2020}. We use geometric rounding so that we only have to deal with countably many possible job sizes. Now, jobs of any given size $p$ form a job class $\CC=\CC_p$. Using \Cref{pro.sample}, we can relate their actual cardinality $n_\CC$ with the $\varphi$-estimate $n_{\CC,\varphi}$. Putting everything together yields the Load Lemma, which compares the load $L$ and the load estimate $L_\varphi$. The lemma relies intrinsically on the lower bound~$R_\mathrm{low}$ for $R(\CJ)$. Consider a job set $\CJ$ like the one in \Cref{fig:counterx}, only one job carries all the load while there are lots of other jobs with size zero (or negligible size $\epsilon>0$). Then $R(\CJ)=\frac{1}{m}$ and a statement as in \Cref{Loadlemma} could not be true for $\varepsilon< \min(1,\varphi^{-1}-1)$ since $L_\varphi\in\{0,\varphi^{-1}L\}$.

\Cref{fig.ratios} shows the behavior of the average load on three randomly chosen permutations of a classical input sequence. As predicted, this average load approaches a straight line for large number of machines. The Load Lemma is an important theoretical tool but only provides asymptotic guarantees. In \Cref{sec.Lunderest} we explore practical guarantees for small numbers of machines.

\subsection{A simple $\boldsymbol{1.75}$-competitive algorithm}\label{sec.1.75}
We modify the semi-online algorithm $\mathrm{Light Load}$ from the literature to obtain a very simple nearly $1.75$-competitive algorithm.
For any $0\le t\le n$ let $M_\midM^{t}$ be a machine having the $\lfloor m/2 \rfloor$-lowest load at time $t$, i.e.\ right before job $J_{t+1}$ is scheduled. Let $l_\midM^{t}$ be its load and let $l_\lowM^{t}$ be the smallest load of any machine.
We recall the algorithm $\LL{L_\guess}$ from Albers and Hellwig~\cite{albers_semi-online_2012}, where the parameter $L_\guess$ is a guess for $L$.

\begin{algorithm}[H]
\caption{The (semi-online) algorithm $\LL{L_\guess}$~\cite{albers_semi-online_2012}.}\label{alg.exceptional2}
\begin{algorithmic}[1]
\State \textit{Let $J_t$ be the job to be scheduled and let $p_t$ be its size.}
\If{$l_\lowM^{t-1} \le 0.25  L_\guess$ \textbf{or} $l_\midM^{t-1} + p_t > 1.75 L_\guess$}
\State Schedule $J_t$ on any least loaded machine;
\Else
{ }schedule $J_t$ on $M_\midM^{t-1}$;
\EndIf
\end{algorithmic}
\end{algorithm}

$\LL{L}$ has been analyzed in the setting where the average load $L$ is known in advance, i.e.\ with fixed parameter $L_\guess=L$. Albers and Hellwig obtain the following:
\begin{theorem}[\cite{albers_semi-online_2012}]\label{te.ll}
$\LL{L}$ is adversarially $1.75$-competitive, i.e.\ for every job sequence $\CJ^\sigma$ with average load $L=L[\CJ]$ there holds $\LL{L}(\CJ^\sigma)\le 1.75\OPT(\CJ)$.
\end{theorem}

The proof from~\cite{albers_semi-online_2012} is complicated and not repeated in this paper. We need to deal with more general guesses $L_\guess$ that are slightly off. The following corollary is derived from \Cref{te.ll} by enlarging the input sequence.

\begin{restatable}{corollary}{coll}\label{co.ll}
Let $\CJ^\sigma$ be any (ordered) input sequence and let $L_\guess \ge L[\CJ]$. Then the makespan of $\LL{L_\guess}$ is at most $1.75 \cdot \max (L_\guess,\OPT(\CJ))$.
\end{restatable}

The idea of the proof is rather simple. We can add jobs to the end of the sequence $\CJ^\sigma$ such that for the resulting sequence $\CJ'^{\sigma'}$ there holds $\OPT(\CJ')=\max(L_\guess,\OPT(\CJ))$. We then apply \Cref{te.ll} to see that $\LL{L_\guess}$ has cost at most $1.75 \cdot \max (L_\guess,\OPT(\CJ))$ on this sequence. Passing over to the prefix $\CJ^\sigma$ of $\CJ'^{\sigma'}$ cannot increase this cost. A technical proof is left to \Cref{sec.1.75.p} for completeness.

We also need to deal with guesses $L_\guess$ that are totally of. Since $\LL{L_\guess}$  only considers the least or the $\lfloor m/2\rfloor$-th least loaded machine we get by \Cref{prop1}:
\begin{corollary}\label{co.ll2}
For any (ordered) sequence $\CJ^\sigma$ and any value $L_\guess$ the makespan of $\LL{L_\guess}$ is at most $(1+2R(\CJ))\OPT(\CJ)$. In particular, it is at most $3\OPT(\CJ)$.
\end{corollary}

\subsubsection*{Adapting LightLoad to the random-order model}
Let $\delta=\delta(m)=1/\log(m)$ be the \emph{margin of error our algorithm allows}. We will see that our algorithm is $(1.75+O(\delta))$-competitive. In fact, any function with $\delta(m)\in \omega(m^{-1/4})$ and $\delta(m)\in o_m(1)$ would do. Given an input sequence $\CJ^\sigma$ let $\hat L_\pre=\hat L_\pre[\CJ^\sigma]=\frac{L_{1/4}[\CJ^\sigma]}{1-\delta}$ be our \emph{guess for $L$}. We use the index 'pre' since our main algorithm later will use a slightly different guess $\hat L$. In this section we consider the algorithm $\mathrm{Light Load ROM}=\LL{\hat L_\pre}$. Let us observe first that this is indeed an online algorithm, not only a semi-online algorithm as one might expect since the $\textbf{if}$-clause uses the guess $\hat L_\pre$ before it is known.

\begin{lemma}\label{le.llromon}
The algorithm $\mathrm{Light Load ROM}$ can be implemented as an online algorithm.
\end{lemma}

\begin{proof}
It suffices to note that the $\textbf{if}$-clause always evaluates to \textsc{true} for $t<n/4$, i.e.\ before $L_\guess=\hat L_\pre$  is known. Indeed, in this case $l_\lowM \le 0.25L_{1/4} < 0.25  \hat L_\pre$ by \Cref{le.avglb}.
\end{proof}

We now prove the main theorem. \Cref{co.1.75} follows immediately from \Cref{le.nearly}.

\begin{theorem}\label{te.1.75}
The algorithm $\mathrm{Light Load ROM}$ is nearly $1.75$-competitive.
\end{theorem}
\begin{corollary}\label{co.1.75}
$\mathrm{Light Load ROM}$ is $1.75$-competitive in the secretary model for $m\rightarrow\infty$.
\end{corollary}

\begin{figure}[H]
    \centering
    \resizebox{\textwidth}{!}{
\begin{tikzpicture}
   \draw (0,0) rectangle (12,2) node[midway] {stable} ;
   \draw[pattern=north west lines, pattern color=black] (12,0) rectangle (12.3,2);
   \draw[dashed, line width=1.5mm, rounded corners=20](0.2,0.3) rectangle (12.1,1.8);
   \draw [->] (3.65,1.7) to[out=-20, in=0 ,looseness=1.6]  (3.55,1)   to[out=180,in=200,looseness=1.6]  (3.45,1.7) ;
    \node at (3.55,1.25){$S_n$};

   \draw (12.3,0) rectangle (14,2)node[midway] {simple} ;
   \draw[dashed, line width=1.5mm, rounded corners=15](12.5,0.3) rectangle (13.8,1.8);
    \draw [->] (13.25,1.7) to[out=-20, in=0,looseness=1.5]  (13.15,1.2)   to[out=180,in=200,looseness=1.5]  (13.05,1.7) ;
    \node at (13.15,1.43){$S_n$};
\end{tikzpicture}
    }
    \caption{The lay of the land of the analysis.
The algorithm is $(c+O(\delta))$-competitive on simple and proper stable sequences. Only the small unstable remainder (hashed) is problematic. Dashed lines mark orbits under the action of the permutation group $S_n$. Simple sequences stay simple under permutation. Non-simple orbits have at most an $\delta$-fraction, which is unstable (hashed). Thus, the algorithm is $(c+O(\delta))$-competitive with probability at least $1-\delta$ after random permutation.}
    \label{fig:simplestable}
    \end{figure}
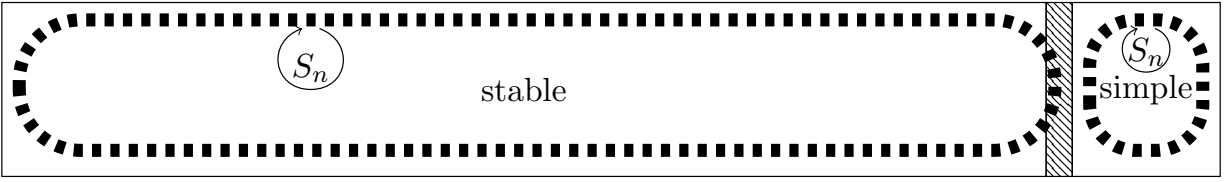

\begin{proof}[Proof of \Cref{te.1.75}]Our analysis forms a triad, which outlines how we are going to analyze our more sophisticated $1.535$-competitive algorithm later on as well. 

\noindent{\bf Analysis basics:} By \Cref{co.ll2} the algorithm $\mathrm{Light Load ROM}$ is adversarially $3$-competitive. We call the input set $\CJ$ \emph{simple} if $|\CJ|\le m$ or $R[\CJ]<\frac{3}{8}$. If $|\CJ|\le m$ every job is scheduled onto an empty machine, which is optimal. If $R[\CJ]<\frac{3}{8}$, \Cref{co.ll2} bounds the competitive ratio by $1+2R[\CJ]<1.75$. We thus are left to consider non-simple, so called \emph{proper}, job sets.

\noindent{\bf Stable job sequences:} We call a sequence $\CJ^\sigma$ \emph{stable} if $L\le \hat L_\pre \le \frac{1+\delta}{1-\delta}L$ holds true. By the Load Lemma, \Cref{Loadlemma}, the probability of the sequence $\CJ^\sigma$ being stable is at least $1-\delta$ if we choose $m$ large enough and $\CJ$ proper. Here we use that $\delta(m) = 1/\log(m) \in \omega(m^{-1/4})$.

\noindent{\bf Adversarial Analysis:} By \Cref{co.ll}, the makespan of $\mathrm{Light Load ROM}$ on stable sequences is at most
$1.75 \max (\hat L_\pre(\CJ),\OPT(\CJ))\le 1.75\frac{1+\delta}{1-\delta}\OPT(\CJ)=\big(1.75+\frac{3.5\cdot\delta}{1-\delta}\big)\OPT(\CJ).$

\noindent{\bf Conclusion:} Let $\varepsilon>0$. Since $\delta(m)\rightarrow 0$, we can choose $m$ large enough such that $\frac{3.5\delta(m)}{1-\delta(m)}\le \varepsilon$. In particular $\bP_{\sigma\sim S_n} [\mathrm{Light Load ROM}({\cal J}^\sigma) \geq (c+\varepsilon) OPT({\cal J})] \leq \delta(m) \le \varepsilon$ since the only sequences where the inequality does not hold are proper but not stable. This concludes the second condition of nearly competitivity.
\end{proof}

\subsubsection*{Why underestimating $L$ is actually not as bad as one may think.}\label{sec.Lunderest}
So far we were careful to choose our guess $\hat L_\mathrm{pre}$ in such a way that it is unlikely to underestimate~$L$ since this allowed us to prove results in a self-contained fashion, using \Cref{te.ll} from~\cite{albers_semi-online_2012} only as a black box. One should note that their analysis also allows us to tackle guesses $L_\guess<L$.

\begin{lemma}\label{le.underest}
Let $\CJ^\sigma$ be any (ordered) input sequence and let $L_\guess = (1-\delta)L[\CJ]$ for some $\delta \ge 0$. Then the makespan of $\LL{L_\guess}$ is at most $1.75(1+\delta) \OPT(\CJ)$.
\end{lemma}

Showing this lemma requires carefully rereading the analysis of Albers and Helwig~\cite{albers_semi-online_2012}. 
We next describe how their analysis has to be adapted to derive \Cref{le.underest}.

\begin{proof}[How to adapt the proof from~\cite{albers_semi-online_2012}]
Consider any input sequence $J_1,\ldots J_n$. Using induction we may assume the result of the lemma to hold on the prefix $J_1,\ldots,J_{n-1}$. In~\cite{albers_semi-online_2012} they argue that the algorithm remains $1.75$-competitive if the least loaded machine had load at most $0.75L$ upon arrival of $J_n$. By a similar reasoning the less strict statement of \Cref{le.underest} holds if the least loaded machine had load at most $0.75L+1.75\delta L$ at that time. Thus we are left to consider the case that its load is  $0.75L+1.75\delta L +\epsilon L$ for some $0<\epsilon<0.25-1.75\delta$. Following the arguments~\cite{albers_semi-online_2012}, it suffices to show that every machine received a job of size $0.5L+\epsilon L$. The statement of Lemma 1 in~\cite{albers_semi-online_2012} needs to be weakened to 'At time $t_{j_0}$ the least loaded machine had load at most $(0.25+1.75\delta)L$.' The proof of the lemma remains mostly the same. The only change occurs in the induction step. Here, the size of a job causing a machine to reach load $0.75L+\delta L +\epsilon L$ and, in addition, the corresponding decrease in potential is only $(1-1.75\delta)L$. Similarly, the statement of Lemma 2 needs to be refined to 'the $j_0$-th least loaded machine had load at most $(1.25+1.75\delta)L-\epsilon L=1.75L_\guess-0.5L-\epsilon L$.' The proof of Lemma 2 stays the same. Using these modifications, the rest of the analysis of~\cite{albers_semi-online_2012} can be applied to conclude the proof. 
\end{proof}

\begin{theorem}\label{th.fulllightloadcapM}
Let $\CJ^\sigma$ be any (ordered) input sequence. The makespan of $\mathrm{Light Load ROM}$ on  $\CJ^\sigma$ is $1.75\cdot\left(1+\frac{|\hat L_\pre[\CJ^\sigma]-L|}{L}\right)\OPT$.
\end{theorem}

\begin{proof}
On input permutation $\CJ^\sigma$ the competitive ratio of our algorithm is at most $1.75+\frac{\left|L_{1/4}-L\right|}{L}$ by \Cref{co.ll}  if $L_{1/4} \ge L$. Else, recall that $\OPT\ge L$ and apply \Cref{le.underest}.
\end{proof}

Let us assume for simplicity that the input length $n$ is divisible by $4$. We can always add up to three jobs of size $0$ to obtain such a result. This 'adding' can be simulated by an online algorithm. 
Recall that the \emph{absolute mean deviation} of a random variable $X$ that has nonzero expectation is defined as $\mathrm{MD}[X]=\bE\left[\left|X-\bE[X]\right|\right]$ and its \emph{normalized absolute mean deviation} is $\mathrm{NMD}[X]=\frac{\bE\left[\left|X-\bE[X]\right|\right]}{\bE[X]}$. In particular, $\mathrm{NMD}[L_{1/4}]=\frac{\bE_{\sigma\sim S_N}\left[\left|L_{1/4}-L\right|\right]}{L}.$
From the previous theorem we obtain:


\begin{theorem}\label{th.reduc}
On input set $\CJ$ the competitive ratio of $\mathrm{Light Load ROM}$ in the random-order model is at most $1.75(1+\mathrm{NMD}(\hat L_\pre))$. If $n\le m$ or $R(\CJ)\le \frac{3}{8}$, then $\LL{L_{1/4}}$ is already $1.75$-competitive. 
\end{theorem}
\begin{proof}
The first statement follows from \Cref{th.fulllightloadcapM} by taking expected values. If $n\le m$, the algorithm $\LL{L_{1/4}}$ places every job on a separate machine and is thus optimal. If $R(\CJ)\le \frac{3}{8}$, it is $1.75$-competitive by \Cref{co.ll2}.
\end{proof}

We will now provide estimates on $\mathrm{NMD}(\hat L_\pre)$. \Cref{fig.NAMD} depicts practical estimates, while our analysis will focus on theoretical bounds.

\begin{figure}[H]
\includegraphics[width=\textwidth]{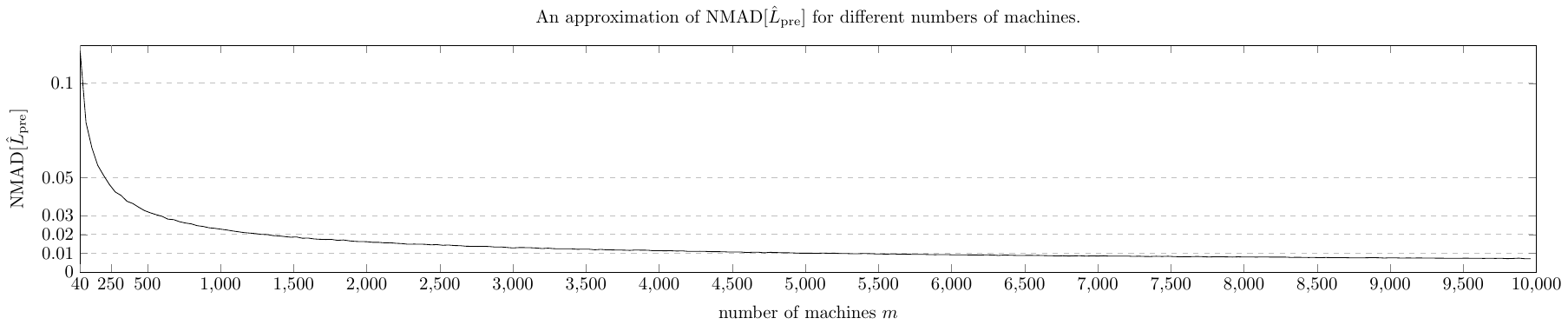}\vspace{-20pt}
\caption{The extra cost for small numbers of machines.  The graph shows an estimation of $\mathrm{NMD}(\hat L_\pre)$ on the lower bound sequence from \cite{albers_better_1999} based on $10,000$ random samples. The curve indicates good performance of $\mathrm{Light Load ROM}$ in practice by \Cref{th.reduc}. 
 }\label{fig.NAMD}
\end{figure}

Given any job set $\CJ$ of size $n>m$ and $R(\CJ) \ge \frac{3}{8}$ we are left to estimate this normalized standard deviation of $L_{1/4}$. One observes that $\mathrm{NMD}[L_{1/4}]$ does not change if we scale all jobs by a common factor $\lambda>0$. By choosing $\lambda=L^{-1}$ we may wlog.\ assume that $L=1$. In particular, $\mathrm{NMD}[L_{1/4}]=\mathrm{MD}[L_{1/4}]$. Now $R(\CJ)\ge \frac{3}{8}$ implies that all jobs have size at most $8/3$. The following lemma allows us to reduce ourselves to particularly easy job instances.

\begin{lemma}
Consider two jobs $J_a,J_b\in \CJ$ of sizes $p_a\le p_b$ and $0\le\varepsilon\le p_a$. If we set the size of $J_a$ to $p_a-\varepsilon$ and the size of $J_b$ to $p_b+\varepsilon$, then $
\mathrm{MD}(L_{1/4})$ does not decrease.
\end{lemma}

\begin{proof}
Consider $\mathrm{MD}(L_{1/4})(p_1,\ldots,p_n)=\bE[|L_{1/4}-L|](p_1,\ldots,p_n)$ as a function on the job sizes $p_1,\ldots,p_n$. This function is convex since it is a convex combination of the convex functions $|L_{1/4}[\sigma]-L|$ for all~$\sigma\in S_n$.
\end{proof}

We apply the previous lemma to any pairs of $J_a,J_b\in \CJ$ of sizes $0<p_a\le p_b<8/3$ to set either $p_a=0$ or $p_b=8/3$. We then repeat this process till all jobs but one last one have either size $0$ or $8/3$. So far at most $\left\lfloor\frac{3m}{8}\right\rfloor$ jobs have size~$8/3$ since by assumption $L=1$. Using again the fact that $\mathrm{MD}[L_{1/4}]$ is convex in the size of jobs, setting the size $p$ of this remaining job to at least one of the values $0$ or $8/3$ cannot decrease $\mathrm{MD}[L_{1/4}]$. Let us do so. This breaks the assumption that $L=1$, which is why we consider $\mathrm{MD}$ instead of $\mathrm{NMD}$. Let $K=K(n)$ be the number of jobs of size~$8/3$. Then $K$ is either $\left\lfloor\frac{3m}{8}\right\rfloor$ or $\left\lfloor\frac{3m}{8}\right\rfloor+1$. Let $X\sim\mathrm{HyperGeom}(n,K,n/4)$, in other words $X$ corresponds to drawing $n/4$ elements without replacement from a population of size $n$ that contains precisely $K$ successes. Then $L_{1/4}=\frac{8X}{3m}$ for our modified job set. Since all modifications never caused $\mathrm{MD}[X]$ to decrease we have shown so far:

\begin{lemma}\label{le.Ldevbound}
Let $\CJ$ be a job set of size $n>m$ with $R(\CJ)\le \frac{3}{8}$. Then we can choose $K$ either $\left\lfloor\frac{3m}{8}\right\rfloor$ or $\left\lfloor\frac{3m}{8}\right\rfloor+1$ such that for $X\sim\mathrm{HyperGeom}(n,K,n/4)$ there holds $\mathrm{NMD}[L_{1/4}]\le \frac{8}{3m}\mathrm{MD}[X]$.
\end{lemma}

It is possible to evaluate the standard mean deviation of $X\sim\mathrm{HyperGeom}(n,K,n/4)$ directly using the techniques in~\cite{diaconis1991closed}. Since such an analysis is quite complex we present a simpler proof, which yields somewhat worse bounds.

\begin{lemma}\label{le.XYcomp}
Let $X\sim\mathrm{HyperGeom}(n,K,n/4)$ and $Y\sim\mathrm{Bin}(K,1/4)$ then $\mathrm{MD}(X)\le\mathrm{MD}(Y)$.
\end{lemma}

\begin{proof}
Indeed, both random variables correspond to $K$ draws from a population of size $n$ that contains $n/4$ successes. For $X$ these draws occur without replacement, while for $Y$ these are draws with replacement. Let $X_i$ respectively $Y_i$ be the respective random variable, which only considers the first $i$ draws for $0\le i\le K$. We can show via induction that  the random variable $|Y_i-i/4|$ dominates $|X_i-i/4|$ by a case distinction on the possible values of $Y_{i-1}-(i-1)/4$ and $X_{i-1}-(i-1)/4$. Thus, the random variable $|Y-\bE[Y]|$ dominates $|X-\bE[X]|$. This implies $\mathrm{MD}(Y)=\bE[|Y-\bE[Y]|]\ge\bE[|X-\bE[X]|]=\mathrm{MD}(X)$.
\end{proof}

Given $Y\sim\mathrm{Bin}(K,p)$, we are interested in $p=1/4$, let $\mathrm{bin}(k,K,p)=P[Y=k]={K\choose k}p^k(1-p)^{K-k}$. Then we can evaluate the median deviation of $Y$ using de Moivre's theorem. 

\begin{theorem}[de Moivre]
$\mathrm{MD}(Y)=2\lfloor pK+1\rfloor (1-p)\mathrm{bin}(\lfloor pK+1\rfloor,K,p)$ for $Y\sim\mathrm{Bin}(K,p)$.
\end{theorem}

A proof of the theorem can be found in \cite{diaconis1991closed}.
Consider $Y\sim\mathrm{Bin}(K,1/4)$ and set $k=\lfloor K/4+1\rfloor$.
Using Stirling's Approximation we derive
\begin{align*}
\mathrm{MD}(Y)&=\lfloor 2k\rfloor \cdot 3/4 \cdot \mathrm{bin}\left(k,K,1/4\right)\\
&=(1+o_K(1)) \frac{3}{8}K \frac{K!}{k!(K-k)!} (1/4)^k(3/4)^{K-k} \\
&=(1+o_K(1)) \frac{3}{8}K \frac{\sqrt{2\pi K}\left(\frac{K}{e}\right)^K}{\sqrt{2\pi k}\left(\frac{k}{e}\right)^k\sqrt{2\pi (K-k)}\left(\frac{K-k}{e}\right)^{K-k}} \left(\frac{1}{4}\right)^k\left(\frac{3}{4}\right)^{K-k} \\
&=(1+o_K(1)) \frac{3}{8}K \sqrt\frac{K}{2\pi k(K-k)} \left(\frac{K}{4k}\right)^k\left(\frac{3}{4}\frac{K}{K-k}\right)^{K-k}  .
\end{align*}
Recall $K/4<k\le K/4+1$. Thus $\left(\frac{K}{4k}\right)^k\le 1$ and $\left(\frac{3}{4}\frac{K}{K-k}\right)^{K-k} \le \left(1+\frac{1}{3K/4-1}\right)^{3K/4}=e+o_K(1)$. 
We get that
$\mathrm{MD}(Y)\le (1+o_K(1)) \frac{3e}{8}K \sqrt\frac{K}{2\pi k(K-k)} $. Recall that $K=\left\lfloor\frac{3m}{8}\right\rfloor$ or $K=\left\lfloor\frac{3m}{8}\right\rfloor+1$. In particular, $K=(1+o_m(1))\frac{8}{3}m$ and $k=(1+o_m(1))\frac{2}{3}m$ . Thus 
\[\mathrm{MD}(Y)= (1+o_m(1)) \frac{3e}{8}\frac{8m}{3} \sqrt\frac{8m/3}{2\pi 2m/3\cdot(1m/3)} = (1+o_m(1))\sqrt\frac{6e^2m}{\pi}.\]

Consider any job set $\CJ$ of size $n$ with $R(\CJ)\le \frac{3}{8}$. Combining the previous bound with \Cref{le.Ldevbound} and \Cref{le.XYcomp} yields
\[\mathrm{NMD}[L_{1/4}]\le \frac{8}{3m}\mathrm{MD}[X] =  (1+o_m(1))\frac{8}{3m}\sqrt\frac{6e^2m}{\pi}<\frac{10.02+o_m(1)}{\sqrt{m}}.\]

This bound allows to establish a competitive ratio in the random-order model.

\begin{theorem}\label{th.fulllightloadcap2}
The competitive ratio of $\mathrm{Light Load ROM}$ in the random-order model is $1.75+\frac{18}{\sqrt{m}}+O\left(\frac{1}{m}\right)$.
\end{theorem}
\begin{proof}
This is a consequence of \Cref{th.reduc} and the prior bound on $\mathrm{NMD}[L_{1/4}]$.
\end{proof}

\begin{remark}
The constant in the previous theorem is far from optimal. As mentioned before a first improvement can be derived by estimating the absolute mean average deviation of the hypergeoemtric distribution directly using the techniques from~\cite{diaconis1991closed}. A much stronger improvement results from filtering out huge jobs reducing essentially to the case that $R(\CJ)=1$. Given any guess $L_\guess$ let $\delta(L_\guess)$ be $L-L_\guess$ if $L_\guess \le L$; $L_\guess-\OPT$ if $L_\guess>\OPT$; and $0$ else. Combining \Cref{co.ll}  and \Cref{le.underest} yields that $\LL{L_\guess}$ is $1.75(1+\delta(L_\guess))$-competitive. So far, we picked for $L_\guess$ an estimator for $L$. But besides $L$ we could also try to estimate the lower bound for $\OPT$, that is we could consider $B_\mathrm{pre}[\CJ^\sigma]=\max\left(L_{1/4}[\CJ^\sigma],p^{n/4}_\jmax[\CJ^\sigma]\right)$, which is similar to how we estimate $\OPT$ in our main algorithm. The nature of the guess $B_\mathrm{pre}[\CJ^\sigma]$ ensures that only few job have size exceeding $B_\mathrm{pre}$; only $4$ in expectation. A more careful analysis reveals that the competitive ratio of $\LL{B_\mathrm{pre}}$ is in fact $1.75+\frac{4.4}{\sqrt{m}}+\frac{7}{m}+O\left(m^{-3/2}\right)$ in the random-order model.
\end{remark}

\section{The new, nearly 1.535-competitive algorithm}\label{sec.algorithm}

\begin{figure}[t]
\resizebox{\textwidth}{!}{
\begin{tikzpicture}
\draw[draw=blue!80, very thick] (2.8,1) rectangle node[above= -1cm] { \resizebox{6.3cm}{!}{\begin{tikzpicture}[
smallstyle/.style={fill=black!80},
mediumstyle/.style={fill=gray!70},
]

\draw[draw=none] (0.0,0) -- (-0.4,0);

\draw[smallstyle] (0.0,0) rectangle (0.3,0.498221250635);


\draw[smallstyle] (0.3,0) rectangle (0.6,0.328120246241);


\draw[smallstyle] (0.6,0) rectangle (0.9,0.29360077356);


\draw[smallstyle] (0.9,0) rectangle (1.2,0.24537432934);


\draw[smallstyle] (1.2,0) rectangle (1.5,0.122291033642);


\draw[smallstyle] (1.5,0) rectangle (1.8,0.0995934755325);


\draw[smallstyle] (1.8,0) rectangle (2.1,0.0992786523462);


\draw[smallstyle] (2.1,0) rectangle (2.4,0.0972837876892);


\draw[smallstyle] (2.4,0) rectangle (2.7,0.0894246039287);


\draw[smallstyle] (2.7,0) rectangle (3.0,0.089080993847);


\draw[smallstyle] (3.0,0) rectangle (3.3,0.0887643258615);


\draw[smallstyle] (3.3,0) rectangle (3.6,0.0865657914244);


\draw[smallstyle] (3.6,0) rectangle (3.9,0.0833225793393);


\draw[smallstyle] (3.9,0) rectangle (4.2,0.0793138962095);


\draw[smallstyle] (4.2,0) rectangle (4.5,0.0724484688372);


\draw[smallstyle] (4.5,0) rectangle (4.8,0.07029300532);


\draw[smallstyle] (4.8,0) rectangle (5.1,0.0678367006483);


\draw[smallstyle] (5.1,0) rectangle (5.4,0.066689467333);


\draw[smallstyle] (5.4,0) rectangle (5.7,0.0654076306032);


\draw[smallstyle] (5.7,0) rectangle (6.0,0.0580182887605);


\draw[smallstyle] (6.0,0) rectangle (6.3,0.0540405846805);


\draw[smallstyle] (6.3,0) rectangle (6.6,0.0504310521831);


\draw[smallstyle] (6.6,0) rectangle (6.9,0.0502308191398);


\draw[smallstyle] (6.9,0) rectangle (7.2,0.0400406246461);


\draw[smallstyle] (7.2,0) rectangle (7.5,0.0356719776493);


\draw[smallstyle] (7.5,0) rectangle (7.8,0.0337437367067);


\draw[smallstyle] (7.8,0) rectangle (8.1,0.031119714696);


\draw[smallstyle] (8.1,0) rectangle (8.4,0.0300256923413);


\draw[smallstyle] (8.4,0) rectangle (8.7,0.0294381205328);


\draw[smallstyle] (8.7,0) rectangle (9.0,0.0288727636286);


\draw[smallstyle] (9.0,0) rectangle (9.3,0.0285877811619);


\draw[smallstyle] (9.3,0) rectangle (9.6,0.0264442926506);


\draw[smallstyle] (9.6,0) rectangle (9.9,0.0224346768974);


\draw[smallstyle] (9.9,0) rectangle (10.2,0.0203136575482);


\draw[smallstyle] (10.2,0) rectangle (10.5,0.0153039440602);


\draw[smallstyle] (10.5,0) rectangle (10.8,0.0120382897926);


\draw[smallstyle] (10.8,0) rectangle (11.1,0.00864442199583);


\draw[smallstyle] (11.1,0) rectangle (11.4,0.00524726489768);
\draw[mediumstyle] (11.1,0.5) rectangle (11.4,0.828120246241);


\draw[smallstyle] (11.4,0) rectangle (12.6,0);
\draw [decorate,decoration={brace,amplitude=5pt},xshift=0pt,yshift=0pt](11.4,0) -- (0,0) node [anchor=north, midway, yshift=-2.5pt] {\scriptsize principal machines};
\draw [decorate,decoration={brace,amplitude=5pt},xshift=0pt,yshift=0pt](12.6,0) -- (11.4,0) node [anchor=north, midway, yshift=-2.5pt] {\scriptsize reserve machines};

\end{tikzpicture} }} node[above=3mm] {\bf Sampling Phase} (9.2,2.7);
\draw[draw=blue!80, very thick] (-1.2,-1.5) rectangle node[above= -1cm] { \resizebox{6.3cm}{!}{\input{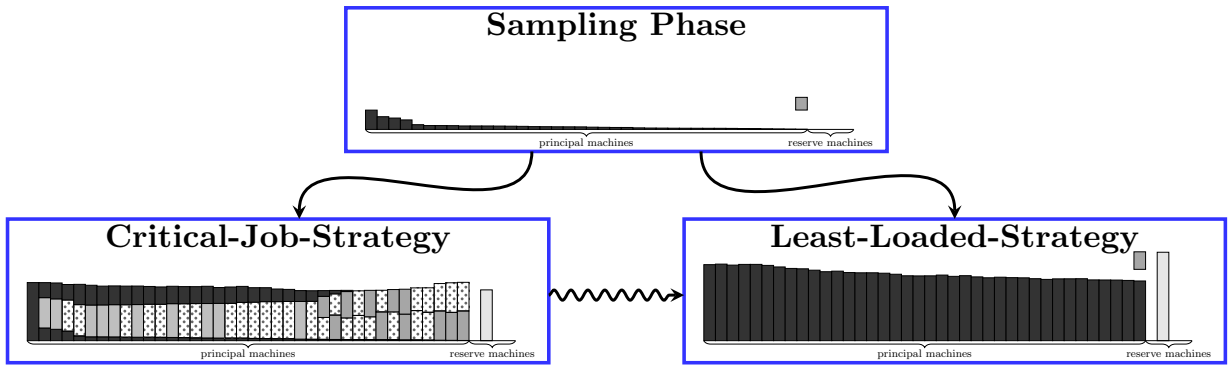}}}  node[above=3mm] {\bf Critical-Job-Strategy}(5.2,0.2);
\draw[draw=blue!80, very thick] (6.8,-1.5) rectangle node[above= -1cm] { \resizebox{6.3cm}{!}{\begin{tikzpicture}[
smallstyle/.style={fill=black!80},
mediumstyle/.style={fill=gray!70},
hugestyle/.style={fill=lightgray!40}
]

\draw[draw=none] (0.0,-0.5) -- (-0.4,-0.5);

\draw[smallstyle] (0.0,-0.5) rectangle (0.3,1.463319039);
\draw[smallstyle] (0.3,-0.5) rectangle (0.6,1.47725494546);
\draw[smallstyle] (0.6,-0.5) rectangle (0.9,1.45166686574);
\draw[smallstyle] (0.9,-0.5) rectangle (1.2,1.46984184715);
\draw[smallstyle] (1.2,-0.5) rectangle (1.5,1.47042872348);
\draw[smallstyle] (1.5,-0.5) rectangle (1.8,1.44396477212);
\draw[smallstyle] (1.8,-0.5) rectangle (2.1,1.40405393062);
\draw[smallstyle] (2.1,-0.5) rectangle (2.4,1.36580565292);
\draw[smallstyle] (2.4,-0.5) rectangle (2.7,1.35624317583);
\draw[smallstyle] (2.7,-0.5) rectangle (3.0,1.32042181436);
\draw[smallstyle] (3.0,-0.5) rectangle (3.3,1.28122451371);
\draw[smallstyle] (3.3,-0.5) rectangle (3.6,1.29551170244);
\draw[smallstyle] (3.6,-0.5) rectangle (3.9,1.25930490857);
\draw[smallstyle] (3.9,-0.5) rectangle (4.2,1.25270698828);
\draw[smallstyle] (4.2,-0.5) rectangle (4.5,1.25331105937);
\draw[smallstyle] (4.5,-0.5) rectangle (4.8,1.24240239286);
\draw[smallstyle] (4.8,-0.5) rectangle (5.1,1.21557464257);
\draw[smallstyle] (5.1,-0.5) rectangle (5.4,1.17592564058);
\draw[smallstyle] (5.4,-0.5) rectangle (5.7,1.17245499017);
\draw[smallstyle] (5.7,-0.5) rectangle (6.0,1.1745476073);
\draw[smallstyle] (6.0,-0.5) rectangle (6.3,1.19416562656);
\draw[smallstyle] (6.3,-0.5) rectangle (6.6,1.1612802245);
\draw[smallstyle] (6.6,-0.5) rectangle (6.9,1.1785354634);
\draw[smallstyle] (6.9,-0.5) rectangle (7.2,1.14732359072);
\draw[smallstyle] (7.2,-0.5) rectangle (7.5,1.12430604002);
\draw[smallstyle] (7.5,-0.5) rectangle (7.8,1.13919274206);
\draw[smallstyle] (7.8,-0.5) rectangle (8.1,1.12784085446);
\draw[smallstyle] (8.1,-0.5) rectangle (8.4,1.12303994167);
\draw[smallstyle] (8.4,-0.5) rectangle (8.7,1.1010920727);
\draw[smallstyle] (8.7,-0.5) rectangle (9.0,1.08263035799);
\draw[smallstyle] (9.0,-0.5) rectangle (9.3,1.09628647979);
\draw[smallstyle] (9.3,-0.5) rectangle (9.6,1.09769634598);
\draw[smallstyle] (9.6,-0.5) rectangle (9.9,1.10021288007);
\draw[smallstyle] (9.9,-0.5) rectangle (10.2,1.06962737214);
\draw[smallstyle] (10.2,-0.5) rectangle (10.5,1.06754131669);
\draw[smallstyle] (10.5,-0.5) rectangle (10.8,1.0652027866);
\draw[smallstyle] (10.8,-0.5) rectangle (11.1,1.05612279349);
\draw[smallstyle] (11.1,-0.5) rectangle (11.4,1.04201033623);
\draw[mediumstyle] (11.1,1.34201033623) rectangle (11.4,1.79509760776);
\draw[hugestyle] (11.7,-0.5) rectangle (12.0,1.78372408847);
\draw[smallstyle] (11.4,-0.5) rectangle (12.6,-0.5);
\draw [decorate,decoration={brace,amplitude=5pt},xshift=0pt,yshift=0pt](11.4,-0.5) -- (0,-0.5) node [anchor=north, midway, yshift=-2.5pt] {\scriptsize principal machines};
\draw [decorate,decoration={brace,amplitude=5pt},xshift=0pt,yshift=0pt](12.6,-0.5) -- (11.4,-0.5) node [anchor=north, midway, yshift=-2.5pt] {\scriptsize reserve machines};
\end{tikzpicture} }} node[above=3mm] {\bf Least-Loaded-Strategy}(13.2,0.2);
\draw[thick,->,>=stealth, line width=1pt] (5,1) to [ out=-90, in=90] (2.25,0.2);
\draw[thick,->,>=stealth, line width=1pt] (7,1) to [ out=-90, in=90] (10,0.2);
\draw[thick,->,>=stealth, line width=1pt, decorate, decoration={snake, amplitude=0.6mm, segment length=2mm, post length=1mm}] (5.2,-0.7) to  (6.8,-0.7);

\end{tikzpicture}}
\caption{The 1.535-competitive algorithm.  First, few jobs are sampled. Then, the algorithm decides between two strategies. The Critical-Job-Strategy tries to schedule critical jobs ahead of time. The Least-Loaded-Strategy follows a greedy approach, which reserves some machines for large jobs. Sometimes, we realize very late that the Critical-Job-Strategy does not work and have to switch to the Least-Loaded-Strategy 'on the fly'. We never switch in the other direction.
}
    \label{fig:inputsequence}
\end{figure}

Our new algorithm achieves a competitive ratio of $c=\frac{1+\sqrt{13}}{3}\approx 1.535$. Let $\delta=\delta(m)=\frac{1}{\log(m)}$ be the \emph{margin of error our algorithm allows}. Throughout the analysis it is mostly sensible to treat $\delta$ as a constant and forget about its dependency on $m$.
Our algorithm maintains a certain set $\CM_\res$ of $\lceil\delta m\rceil$ \emph{reserve machines}. Their complement, the \emph{principal machines}, are denoted by~$\CM$. Let us fix an input sequence $\CJ^\sigma$. Let $\hat L=\hat L[\CJ^\sigma]=L_{\delta^2}[\CJ^\sigma]$. For simplicity, we hide the dependency on $\CJ^\sigma$ whenever possible. Our online algorithm uses $ B=\max\left(p^{\delta^2 n}_\jmax,\hat L\right)$ as an \emph{estimated lower bound for $\OPT$}, which is known after the first $\lfloor\delta^2 n\rfloor$ jobs are treated. Our algorithm uses geometric rounding implicitly. Given a job $J_t$ of size $p_t$ let $f(p_t)=(1+\delta)^{\left\lfloor \log_{1+\delta}p_t\right\rfloor}$ be its \emph{rounded size}. We also call $J_t$ an $f(p_t)$-job. Using rounded sizes, we introduce job classes. Let $p_\jsmall =c-1=\frac{\sqrt{13}-2}{3}\approx 0.535$ and $p_\jbig=\frac{c}{2}=\frac{1+\sqrt{13}}{6}\approx 0.768$. Then we call job $J_t$
\begin{itemize}
\setlength{\parskip}{0pt} \setlength{\itemsep}{0pt plus 1pt}
\item \emph{small} if $f(p_t)\le p_\jsmall B$ and \emph{critical} else,
\item \emph{big} if $f(p_t)>p_\jbig B$,
\item \emph{medium} if $J$ is neither small nor big, i.e.\ $p_\jsmall B\le f(p_t)\le p_\jbig B$,
\item \emph{huge} if its (not-rounded) size exceeds $B$, i.e.\ $B<p_t$, and \emph{normal} else.
\end{itemize}

Consider the sets $\CP_\jmed=\{(1+\delta)^i \mid (1+\delta)^{-1}p_\jsmall B < (1+\delta)^i \le p_\jbig B\}$ and $\CP_\jbig=\{(1+\delta)^i \mid  p_\jbig B < (1+\delta)^i \le  B\}$ corresponding to all possible rounded sizes of medium respectively big jobs, excluding huge jobs. Let $\CP=\CP_\jmed\cup\CP_\jbig$. This subdivision gives rise to a \emph{weight function}, which will be important later. Let $w(p)=1/2$ for $p\in \CP_\jmed$ and $w(p)=1$ for $p\in \CP_\jbig$. The elements $p\in\CP$ define job classes  $\CC_p\subseteq\CJ$ consisting of all $p$-jobs, i.e.\ jobs  of rounded size $p$. By some abuse of notation, we call the elements in $\CP$ 'job classes', too.
Using the notation from \Cref{sec.sampling} we set $n_p=n_{\CC_p}=|\CC_p|$ and $\hat n_{p}=n_{\CC_p,\delta^2}=|\{J_{\sigma(j)}\mid \sigma(j)\le \delta^2 n \land J_{\sigma(j)} \textrm{ is a $p$-job}\}|$. We want to use the values $\hat n_{p}$, which are available to an online algorithm quite early, to estimate the values $n_{p}$, which accurately describe the set of critical jobs. First, $\delta^{-2}\hat n_{p}$ comes to mind as an estimate for $n_p$. Yet, we need a more complicated guess: $c_p=\max\left(\left\lfloor\left(\delta^{-2}\hat n_{p}-m^{3/4}\right)w(p)\right\rfloor,\hat n_{p}\right) w(p)^{-1}$. It has three desirable advantages. First, for every $p\in\CP$ the value $c_p$ is close to $n_p$ with high probability, but, opposed to $\delta^{-2}\hat n_{p}$, unlikely to exceed it. Overestimating $n_p$ turns out to be far worse than underestimating it. Second, $w(p)c_p$ is an integer and, third, we have  $c_p\ge \hat n_{p}w(p)^{-1}$. A fundamental fact regarding the values $(c_p)_{p\in\CP}$ and $B$ is, of course,  that they are known to the online algorithm once $\lfloor\delta^2 n\rfloor$ jobs are scheduled.

\subparagraph*{Statement of the algorithm:}
If there are less jobs than machines, i.e.\ $n\le m$, it is optimal to put each job onto a separate machine. Else, a short sampling phase greedily schedules each of the first $\lfloor \delta^2 n\rfloor$ jobs to the least loaded principal machine $M\in\CM$. Now, the values $B$ and $(c_p)_{p\in\CP}$ are known. Our algorithm has to choose between two strategies, the Least-Loaded-Strategy and the Critical-Job-Strategy, which we will both introduce subsequently. It maintains a variable $\textsc{strat}$, initialized to $\textsc{Critical}$, to remember its choice. If it chooses the Critical-Job-Strategy, some additional preparation is required. It may at any time discover that the Critical-Job-Strategy is not feasible and switch to the Least-Loaded-Strategy but it never switches the other way around.

\begin{algorithm}[H]
\caption{The complete algorithm: How to schedule job $J_t$.}\label{alg.complete}
\begin{algorithmic}[1]
\State \textit{\textsc{strat} is initialized to \textsc{Critical}, $J_t$ is the job to be scheduled.}
\If{$n\le m$}{ }Schedule $J_t$ on any empty machine;
\ElsIf{$t \le \varphi n$}{ }schedule $J_t$ on a least loaded machine in $\CM$;\Comment{\textit{Sampling phase}}
\Else 
\If{we have $t=\lfloor \varphi n\rfloor+1$}
    \If{$\sum_{p\in\CP} w(p)c_p>m$}
    	{ }$\textsc{strat}\gets \textsc{Least-Loaded}$
    \Else
		{ }proceed with the Preparation for the Critical-Job-Strategy (\Cref{alg.preparation});
	\EndIf
\EndIf
\If{$\textsc{strat}=\textsc{Critical}$}
proceed with the Critical-Job-Strategy (\Cref{alg.main});
\Else 
{ }proceed with the Least-Loaded-Strategy (\Cref{alg.exceptional});
\EndIf
\EndIf
\end{algorithmic}
\end{algorithm}

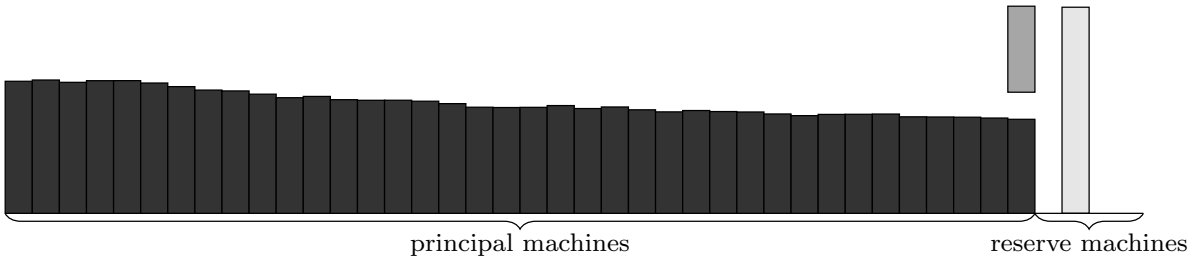
\begin{figure}[b]
    \centering
    \resizebox{\textwidth}{!}{\begin{tikzpicture}[
smallstyle/.style={fill=black!80},
mediumstyle/.style={fill=gray!70},
hugestyle/.style={fill=lightgray!40}
]

\draw[draw=none] (0.0,0) -- (-0.4,0);
\draw[smallstyle]  (0.0,0) rectangle (0.3,1.463319039);
\draw[smallstyle](0.3,0) rectangle (0.6,1.47725494546);
\draw[smallstyle](0.6,0) rectangle (0.9,1.45166686574);
\draw[smallstyle](0.9,0) rectangle (1.2,1.46984184715);
\draw[smallstyle](1.2,0) rectangle (1.5,1.47042872348);
\draw[smallstyle](1.5,0) rectangle (1.8,1.44396477212);
\draw[smallstyle](1.8,0) rectangle (2.1,1.40405393062);
\draw[smallstyle](2.1,0) rectangle (2.4,1.36580565292);
\draw[smallstyle](2.4,0) rectangle (2.7,1.35624317583);
\draw[smallstyle](2.7,0) rectangle (3.0,1.32042181436);
\draw[smallstyle](3.0,0) rectangle (3.3,1.28122451371);
\draw[smallstyle](3.3,0) rectangle (3.6,1.29551170244);
\draw[smallstyle](3.6,0) rectangle (3.9,1.25930490857);
\draw[smallstyle](3.9,0) rectangle (4.2,1.25270698828);
\draw[smallstyle](4.2,0) rectangle (4.5,1.25331105937);
\draw[smallstyle](4.5,0) rectangle (4.8,1.24240239286);
\draw[smallstyle](4.8,0) rectangle (5.1,1.21557464257);
\draw[smallstyle](5.1,0) rectangle (5.4,1.17592564058);
\draw[smallstyle](5.4,0) rectangle (5.7,1.17245499017);
\draw[smallstyle](5.7,0) rectangle (6.0,1.1745476073);
\draw[smallstyle](6.0,0) rectangle (6.3,1.19416562656);
\draw[smallstyle](6.3,0) rectangle (6.6,1.1612802245);
\draw[smallstyle](6.6,0) rectangle (6.9,1.1785354634);
\draw[smallstyle](6.9,0) rectangle (7.2,1.14732359072);
\draw[smallstyle](7.2,0) rectangle (7.5,1.12430604002);
\draw[smallstyle](7.5,0) rectangle (7.8,1.13919274206);
\draw[smallstyle](7.8,0) rectangle (8.1,1.12784085446);
\draw[smallstyle](8.1,0) rectangle (8.4,1.12303994167);
\draw[smallstyle](8.4,0) rectangle (8.7,1.1010920727);
\draw[smallstyle](8.7,0) rectangle (9.0,1.08263035799);
\draw[smallstyle](9.0,0) rectangle (9.3,1.09628647979);
\draw[smallstyle](9.3,0) rectangle (9.6,1.09769634598);
\draw[smallstyle](9.6,0) rectangle (9.9,1.10021288007);
\draw[smallstyle](9.9,0) rectangle (10.2,1.06962737214);
\draw[smallstyle](10.2,0) rectangle (10.5,1.06754131669);
\draw[smallstyle](10.5,0) rectangle (10.8,1.0652027866);
\draw[smallstyle](10.8,0) rectangle (11.1,1.05612279349);
\draw[smallstyle](11.1,0) rectangle (11.4,1.04201033623);
\draw[mediumstyle] (11.1,1.34201033623) rectangle (11.4,2.29509760776);
\draw[hugestyle] (11.7,0) rectangle (12.0,2.28372408847);
\draw[smallstyle](11.4,0) rectangle (12.6,0);
\draw [decorate,decoration={brace,amplitude=5pt},xshift=0pt,yshift=0pt](11.4,0) -- (0,0) node [anchor=north, midway, yshift=-2.5pt] {\scriptsize principal machines};
\draw [decorate,decoration={brace,amplitude=5pt},xshift=0pt,yshift=0pt](12.6,0) -- (11.4,0) node [anchor=north, midway, yshift=-2.5pt] {\scriptsize reserve machines};
\end{tikzpicture} }
    \caption{The Least-Loaded-Strategy schedules jobs greedily. A few machines are reserved for unexpected huge jobs, such as the largest job, which is unlikely to arrive in the sampling phase. 
}
    \label{fig:inputsequence}
\end{figure}

The \textbf{Least-Loaded-Strategy} places any normal job on a least loaded principal machine. Huge jobs are scheduled on any least loaded reserve machine. This machine will be empty, unless we consider rare worst-case orders.

\begin{algorithm}[H]
\caption{The Least-Loaded-Strategy: How to schedule job $J_t$.}\label{alg.exceptional}
\begin{algorithmic}[1]
\If{$J_t$ is huge}{ }schedule $J_t$ on any least loaded reserve machine;
\Else{ }schedule $J_t$ on any least loaded principal machine;
\EndIf
\end{algorithmic}
\end{algorithm}

For the Critical-Job-Strategy we introduce \emph{$p$-placeholder-jobs} for every size $p\in\CP$. Sensibly, the size of a $p$-placeholder-job is $p$.
During the Critical-Job-Strategy we treat placeholder-jobs similar to \emph{real} jobs.
The \emph{anticipated load} $\tilde l_M^t$ of a machine $M$ at time $t$ is the sum of all jobs on it, including placeholder-job, opposed to the common load $l_M^t$, which does not take the latter into account. Note that $\tilde l_M^t$ defines a pseudo-load as introduced  in \Cref{sec.basic}.

\begin{figure}[H]
    \centering
    \resizebox{\textwidth}{!}{\input{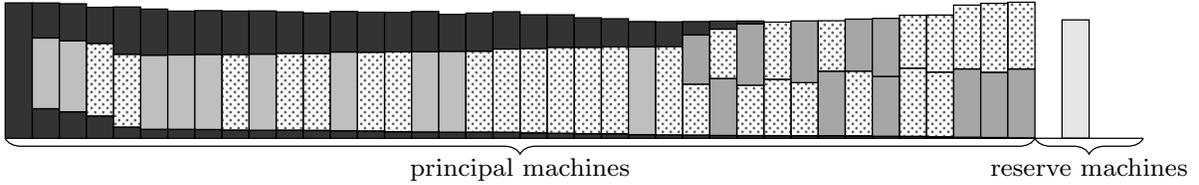}}
    \caption{The Critical-Job-Strategy. Each machine gets either two medium, one large or no critical job. Placeholder jobs (dotted) are assigned during the Preparation and reserve space for critical jobs yet to come. Processing volume of small jobs (dark) 'on the bottom' arrived during the sampling phase. Reserve machines accommodate huge jobs or, possibly, jobs without matching placeholders.
    }
    \label{fig:CriticalJobStrategy}
\end{figure}

During the \textbf{Preparation for the Critical-Job-Strategy} the algorithm maintains a counter $c_p'$ of all $p$-jobs scheduled so far (including placeholders). A job class $p\in\CP$ is called \emph{unsaturated} if $c_p'\le c_p$. First, we add unsaturated medium placeholder-jobs to any principal machine that already contains a medium real job from the sampling phase. We will see in \Cref{le.prep.welldef} that such an unsaturated medium job class always exists. Now, let $m_\mempty$ be the number of principal machines which do not contain critical jobs. We prepare a set $\CJ_\rep$ of cardinality at most $m_\mempty$, which we will then schedule onto these machines. The set $\CJ_\rep$ may contain single big placeholder-jobs or pairs of medium placeholder-jobs. We greedily pick any unsaturated job class $p\in\CP$ and add a $p$-placeholder-job to $\CJ_\rep$. If $p$ is medium, we pair it with a job belonging to any other, not necessarily different, unsaturated medium job class. Such a job class always exists by \Cref{le.prep.welldef}. We stop once all job classes are saturated or if $|\CJ_\rep|=m_\mempty$. We then assign the elements in $\CJ_\rep$ to machines. We iteratively pick the element $e\in\CJ_\rep$ of maximum size and assign the corresponding jobs to the least loaded principal machine, which does not contain critical jobs yet. Sensibly, the size of a pair of jobs in $\CJ_\rep$ is the sum of their individual sizes. We repeat this until all jobs and job pairs in $\CJ_\rep$ are assigned to some principal machine.

\begin{algorithm}[H]
\caption{Preparation for the Critical-Job-Strategy.}\label{alg.preparation}
\begin{algorithmic}[1]
      \While{there is a machine $M$ containing a single medium job}
      \State Add a placeholder $p$-job for an unsaturated size class $p\in\CP_\jmed$ to $M$; $c_p'\gets c_p'+1$;
      \EndWhile
      \While{there is an unsaturated size class $p\in\CP$ and $|\CJ_\rep|< m_\mempty$}
      \State Pick an unsaturated size class $e=p\in \CP$ with $c_p'$ minimal; $w(e)\gets p$;
      $c_p'\gets c_p'+1$;
      \If{$p$ is medium}{ }pick $q\in \CP_\jmed$ unsaturated. $e\gets(p,q)$; $w(e)\gets p+q$; $c_q'\gets c_q'+1$;
 		\EndIf
      \State Add $e$ to $\CJ_\rep$;
      \EndWhile
      
      \While{$\CJ_\rep\neq \emptyset$}
      \State Pick a least loaded machine $M\in\CM$, which does not contain a critical job yet;
      \State Pick $e\in \CJ_\rep$ of maximum size $w(e)$ and add the jobs in $e$ to $M$; \State $\CJ_\rep\gets\CJ_\rep\setminus\{e\}$;
      \EndWhile
\end{algorithmic}
\end{algorithm}
\vspace{-5pt}

\begin{lemma}\label{le.prep.welldef}
In line 2 and 5 of \Cref{alg.preparation} there is always an unsaturated medium size class available. Thus, \Cref{alg.preparation}, the Preparation for the Critical-Job-Strategy, is well defined.
\end{lemma}

\begin{proof}
Concerning line 2, there are precisely $\sum_{p\in\CP_\mathrm{med}} \hat n_p$ machines with critical jobs while there are at least $\sum_{p\in\CP_\mathrm{med}} (c_p-\hat n_p) \ge \sum_{p\in\CP_\mathrm{med}} \hat n_p$ placeholder-jobs available to fill them. Here we make use of the fact that for medium jobs $p\in\CP_\mathrm{med}$ we have $c_p\ge \hat n_p  w(p)^{-1}=2\hat n_p$.

Concerning line 5, observe that so far every machine and every element in $\CJ_\rep$ contains an even number of medium jobs. If the placeholder picked in line 4 was the last medium job remaining, $\sum_{p\in\CP_\mathrm{med}}c_p$ would be odd. But this is not the case since every $c_p$ for $p\in\CP_\mathrm{med}$ is even.
\end{proof}

After the Preparation  is done, the \textbf{Critical-Job-Strategy} becomes straightforward. Each small job is scheduled on a principal machines with least anticipated load, i.e.\ taking placeholders into account. Critical jobs of rounded size $p\in\CP$ replace $p$-placeholder-jobs whenever possible. If no such placeholder exists anymore, critical jobs are placed onto the reserve machines. Again, we try pair up medium jobs whenever possible. If no suitable machine can be found among the reserve machines, we have to switch to the Least-Loaded-Strategy.
We say that the algorithm \emph{fails} if it ever reaches this point. In this case, it should rather have chosen the Least-Loaded-Strategy to begin with. Since all reserve machines are filled at this point, the Least-Loaded-Strategy is impeded, too. The most difficult part of our analysis shows that, excluding worst-case orders, this is not a problem on job sets that are prone to cause failing.
 
\vspace{-5pt}
\begin{algorithm}[H]
\caption{The Critical-Job-Strategy.}\label{alg.main}
\begin{algorithmic}[1]
    \If{$J_t$ is medium or big}
    \textit{let $p$ denote its rounded size;}
    	\If{there is a machine $M$ containing a $p$-placeholder-job $J$}
    	\State Delete the $p$-placeholder-job $J$ and assign $J_t$ to $M$;
    	\ElsIf{$J_t$ is medium and there exists $M\in\CM_\res$ containing a single medium job}
    	\State Schedule $J_t$ on $M$;
    	\ElsIf{there exists an empty machine $M\in\CM_\res$}{ }schedule $J_t$ on $M$;
    	\Else{ }$\textsc{stat}\gets\textsc{Least-Loaded}$; \Comment{We say the algorithm \emph{fails}.}
\State\textbf{use} the Least-Loaded-Strategy (\Cref{alg.exceptional}) from now on; 
    	\EndIf
    \Else
    	 {} assign $J_t$ to the least loaded machine in $\CM$ (take placeholder jobs into account);
    	\EndIf
\end{algorithmic}
\end{algorithm}

\section{Analysis of the algorithm}\label{sec.ana}
\Cref{te.main} is  main result of the paper. \Cref{co.main} follows immediately by \Cref{le.nearly}.

\begin{theorem}\label{te.main}
Our algorithm is nearly $c$-competitive. Recall that $c=\frac{1+\sqrt{13}}{3}\approx 1.535$.
\end{theorem}

\begin{corollary}\label{co.main}
Our algorithm is $c$-competitive in the secretary~model as $m\rightarrow\infty$.
\end{corollary}

The analysis of our algorithm proceeds along the same three reduction steps used in the proof of \Cref{te.1.75}. First, we assert that our algorithm has a bounded adversarial competitive ratio, which approaches $1$ as $R(\CJ)\rightarrow 0$. Not only does this lead to the first condition of nearly competitiveness, it also enables us to introduce \emph{simple} job sets on which we perform well due to basic considerations resulting from \Cref{sec.basic}.
\begin{definition}
A job set $\CJ$ is called \emph{simple} if $R(\CJ)\le \frac{(1-\delta)\delta^3}{2(\delta^2+1)}(2-c)$ or if it consists of at most $m$ jobs. Else, we call it $\emph{proper}$. We call any ordered input sequence $\CJ^\sigma$ \emph{simple} respectively \emph{proper} if the underlying set $\CJ$ has this property.
\end{definition}

Next we are going to sketch out main proof introducing three Main Lemmas. These follow the three proof steps introduced in the proof of \Cref{co.1.75}.

\begin{mainlemma}\label{le.main.basic}
In the adversarial model our algorithm has competitive ratio $4+O(
\delta)$ on general input sequences and $c+O(\delta)$ on simple sequences.

\end{mainlemma}

The proof is discussed later. We are thus reduced to treating \emph{proper} job sets. In the second reduction we introduce \emph{stable} sequences. These have many desirable properties. Most notably, they are suited to sampling. We leave the formal definition to 
\Cref{sec.ana.stable} since it is rather technical. The second reduction shows that stable sequences arise with high probability if one orders a proper job set uniformly randomly.

Formally, for $m$ the number of machines, let $P(m)$ be the maximum probability by which the permutation of any proper sequence may not be stable, i.e.\
\[P(m)=\sup\limits_{\CJ \textrm{ proper}}\bP_{\sigma\sim S_n}\left[\CJ^\sigma \textrm{ is not stable}\right].\]
The second main lemma asserts that this probability vanishes as $m\rightarrow \infty$. 
\begin{mainlemma}\label{le.main.stable}
$\lim\limits_{m\rightarrow\infty}P(m)=0$.
\end{mainlemma}

In other words, non-stable sequences are very rare and of negligible impact in random-order analyses. Thus, we only need to consider stable sequences. In the final, third, step we analyze our algorithm on these. This analysis is quite general. In particular, it does not rely further on the random-order model. Instead, we work with worst-case stable input sequences, i.e.\ we allow the adversary to present any (ordered) stable input sequence.
\begin{mainlemma}\label{le.conclusion}
Our algorithm is adversarially $(c+O(\delta))$-competitive on stable sequences.
\end{mainlemma}

These three main lemmas allow us to conclude the proof of \Cref{te.main}.

\begin{proof}[Proof of \Cref{te.main}]
By \Cref{le.main.basic}, the first condition of nearly competitiveness holds, i.e.\ our algorithm has a constant competitive ratio. Moreover, by \Cref{le.main.basic} and
\Cref{le.conclusion}, given $\varepsilon>0$, we can pick $m_0(\varepsilon)$ such that our algorithm is $(c+\varepsilon)$-competitive on all sequences that are stable or simple if there are at least $m_0(\varepsilon)$ machines. Here, we need that  $\delta(m)\rightarrow 0$ for $m\rightarrow\infty$. This implies that for $m\ge m_0(\varepsilon)$ the probability of our algorithm not being $(c+\varepsilon)$-competitive is at most $P(m)$, the maximum probability with which a random permutation of a proper, i.e.\ non-simple, input sequence is not stable. By \Cref{le.main.stable}, we can find $m(\varepsilon)\ge m_0(\varepsilon)$ such that this probability is less than $\varepsilon$. This satisfies the second condition of nearly competitiveness.
\end{proof}

\subsection{The adversarial case. Proof of \Cref{le.main.basic}}\label{sec.ana.basics}
Recall that the \emph{anticipated load} $\tilde l_M^t$ of a machine $M$ at time $t$ denotes its load including placeholder-jobs. It satisfies the definition of a pseudo-load as introduced in \Cref{sec.basic}. We obtain the following two bounds on the average anticipated load $\tilde L=\sup_t \frac{1}{m} \sum_M \tilde l_M^t$.

\begin{lemma}\label{le.tildeR.bound}
We have $\tilde L \le L+2p_\mathrm{max}$. In particular $\tilde R(\CJ)\le 3$.
\end{lemma}

\begin{proof}
First note that every placeholder-job has at most the size of some job encountered during the sampling phase. In particular,  the size of any placeholder-job is at most $p_\mathrm{max}$. Since there are at most two placeholder-jobs on each machine, the total processing time of all placeholder-jobs is at most $2mp_\mathrm{max}$. The total processing time of real jobs is at most $mL$. Thus the total processing time of all placeholder and real jobs scheduled at any time cannot exceed $m(L+2p_\mathrm{max})$. In particular, at any time $t$, we have $\frac{1}{m} \sum_M \tilde l_M^t\le \frac{1}{m}m(L+2p_\mathrm{max})$. Thus, $\tilde L \le L+2p_\mathrm{max}$. For the second part we conclude that $\tilde R(\CJ)\le \min\left(\frac{\tilde L}{p_\mathrm{max}},\frac{\tilde L}{L}\right) \le \min\left(\frac{L}{p_\mathrm{max}}+2,2\frac{p_\mathrm{max}}{L}+1\right)\le 3.$
\end{proof}

\begin{lemma}\label{le.2ndRtildeBound}
We have $\tilde L \le \left(1+ \frac{1}{\delta^2}\right)L$, in particular $\tilde R(\CJ) \le \left(1+ \frac{1}{\delta^2}\right)R(\CJ)$.
\end{lemma}

Let us first show the following stronger lemma.
\begin{lemma}\label{le.placeholdersize}
The total size of the placeholder-jobs is at most $m\hat L$. In particular, $\tilde L\le L+\hat L$.
\end{lemma}

\begin{proof}[Proof of \Cref{le.placeholdersize}]
For every $p\in\CP$ we schedule at most $c_p\le\delta^{-2}\hat n_p$ placeholder-jobs of type $p$. Thus, the total size of the placeholder-jobs is at most
$\sum_{p\in\CP} \delta^{-2}\hat n_p \le m\hat L.$ The total size of all real jobs is precisely $mL$. Since $\tilde L$ is at most $\frac{1}{m}$-times the total processing time of all jobs, we have $\tilde L \le \frac{1}{m}(mL+m\hat L)$.
\end{proof}

\begin{proof}[Proof of \Cref{le.2ndRtildeBound}]
Observe that $\hat L=L_{\delta^2} \le \delta^{-2} L$. Then, the bound follows from \Cref{le.placeholdersize}.
\end{proof}

We call a machine \emph{critical} if it receives a critical job from the Critical-Job-Strategy but no small job after the sampling phase. Else, we call it \emph{general}. General machines can be analyzed using \Cref{prop1} and~\ref{prop2}. \emph{Critical} machines need more careful arguments.
\begin{lemma}\label{le.basic.bound}
At any time, the load of any general machine is at most $\left(\frac{\tilde R(\CJ)}{1-\delta}+1 +2\delta\right)\OPT(\CJ)$.
\end{lemma}

\begin{proof}
For sequences of length $n\le m$ our algorithm is optimal. Hence assume $n>m$.

During the sampling phase and the Least-Loaded-Strategy, our algorithm always uses either a least loaded machine or a least loaded principal machine. Both lie among the $\lfloor \delta m\rfloor+1$ least loaded machines. By \Cref{prop1} this cannot cause any load to exceed $\left(\frac{m}{m-\lfloor\delta m\rfloor}R(\CJ)+1\right)\OPT(\CJ)\le \left(1+\frac{R(\CJ)}{1-\delta}\right)\OPT(\CJ)$. Observing that $\tilde R(\CJ)\ge R(\CJ)$, see \Cref{le.sensibletildeL}, the lemma holds for every machine that does not receive its last job during the Critical-Job-Strategy.

Now consider a general machine $M$, which received its last job during the Critical-Job-Strategy. Since it is a general machine, it also received a small job during the Critical-Job-Strategy. Let $J$ be the last small job it received. Right before receiving $J$ machine $M$ must have been a principal machine of least anticipated load. In total, it had at most the  $(\lfloor \delta m\rfloor+1)$-smallest anticipated load. By \Cref{prop2} its anticipated load was at most $\left(\frac{\tilde R(\CJ)}{1-\delta}+1\right)\OPT(\CJ)$ after receiving $J$. Afterwards machine $M$ may have received up to two critical jobs, which replaced placeholder-jobs. Since these jobs had at most $(1+\delta)$-times the size of the job they replaced. The load-increase is at most $\delta p_\mathrm{max}\le \delta \OPT$ for each of these two jobs.
\end{proof}

We can now consider critical machines.

\begin{lemma}\label{le.reserve.bound}
The load of a reserve machine is at most
$\min(\max((1+\delta)c B,p_\jmax), 2p_\jmax)$ till it receives a job from the Least-Loaded-Strategy. Critical reserve machines in particular fulfill this condition.
\end{lemma}

\begin{proof}
Every critical reserve machine receives either one big job or at most two medium ones, until the Least-Loaded-Strategy is applied. The second bound, $2p_\jmax$, follows immediately from that. The first bound follows from the fact that a single big job has size at most $p_\mathrm{max}$, while two medium jobs have size at most $2(1+\delta)p_\jbig  B=(1+\delta)c B$. The $(1+\delta)$-factor comes from using rounded sizes in the definition of medium jobs.
\end{proof}

The following lemma uses similar arguments to \Cref{le.nocritjob} only for the adversarial model.

\begin{lemma}\label{le.critical.bound}
The load of a critical machine is at most
$\min((1+\delta)cB+2\frac{\tilde R(\CJ)}{(1-\delta)}\OPT,\frac{L}{1-\delta}+3p_\mathrm{max})$ if it was a principal machine.
\end{lemma}

\begin{proof}
Consider any critical principal machine $M$. Let $J$ be the last job received in the sampling phase. Before $J$ was scheduled on $M$ it was a least loaded principle machine and thus had load at most $\frac{L}{1-\delta}$ by \Cref{prop1}. After $J$ machine $M$ received at most two more jobs and thus its load cannot exceed $\frac{L}{1-\delta}+3p_\mathrm{max}$, the second term in the min-term.

If $J$ was critical, this implies that the load on $M$ of non-critical jobs was at most $\frac{L}{1-\delta}$, while the load of critical jobs no $M$ cannot exceed $2(1+\delta)p_\jbig B =(1+\delta) cB$. The first term in the min-term follows. We are left to consider the case that $M$ did not receive a critical job in the sampling phase, which means that it receives an element of $\CJ_\rep$, else $M$ would not be critical. In fact, assume that $M$ was the $i$-th machine to receive an element from $\CJ_\rep$.

First consider the case $i\le m/2-1$. Right before the while loop in line 7 of \Cref{alg.preparation} machine~$M$ had the $i$-th least anticipated load among the principal machines. By \Cref{le.avglb2} its anticipated load was at most $\frac{m}{m-i-\delta m+1}\tilde L\le \frac{2\tilde L}{1-\delta} \le 2\frac{\tilde R(\CJ)}{1-\delta}\OPT$ before receiving placeholder jobs of processing volume at most $cB$. The processing volume of the placeholder jobs increases by at most a factor $(1+\delta)$ once they are replaced by real jobs. Thus the bound of the lemma follows if $i\le m/2-1$.

Finally, consider the case $i\ge m/2$. Recall that $\delta^{-2} L_{\delta^2} = \hat L \le B$. Since $M$ did not receive a critical job in the sampling phase it follows from \Cref{le.avglb} that its load was at most $\frac{\delta^{-2} L_{\delta^2}}{1-\delta}+(c-1)B \le (c-\frac{\delta}{1-\delta})B\le (1+\delta)\cdot cB$ after the sampling phase. Let $p$ be the processing volume machine $M$ receives from $\CJ_\rep$. Since the algorithm assigns the elements of $\CJ_\rep$ in decreasing order at least $i$ machines received processing volume at least $p$ from $\CJ_\rep$. Thus $i\cdot p \le m\cdot \tilde L$ and, using that $i\ge m/2$, we derive that $p\le \frac{m}{i}\tilde L \le 2\tilde L\le 2\frac{\tilde R(\CJ)}{1-\delta}\OPT$. Again the first term of the min-term follows.
\end{proof}

From these lemmas the two statements of \Cref{le.main.basic} follow.

\begin{corollary}\label{co.goodongeneral}
Our algorithm is adversarially $\left(3+\frac{3}{1-\delta}+2\delta\right)$-competitive.
\end{corollary}

\begin{proof}
	By \Cref{le.tildeR.bound} we have $\tilde R(\CJ)\le 3$,  also recall that $L,p_\mathrm{max}\le\OPT$. By \Cref{le.basic.bound}, \ref{le.reserve.bound} and~\ref{le.critical.bound} the makespan of the algorithm is thus at most $$\max\left(\frac{3}{1-\delta}+1+2\delta,2,\frac{1}{1-\delta}+3\right)\OPT(\CJ)=\left(1+\frac{3}{1-\delta}+2\delta\right)\OPT(\CJ).\qedhere$$
\end{proof}

\begin{corollary}\label{co.goodonsimple}
Our algorithm has makespan at most $(c+2\delta)\OPT$ on simple sequences~$\CJ^\sigma$.
\end{corollary}

\begin{proof}
By \Cref{le.basic.bound}, \ref{le.reserve.bound} and~\ref{le.critical.bound} we see that the makespan of our algorithm is at most \[\max\left(\bigg(\frac{\tilde R(\CJ)}{1-\delta}+1 +2\delta\bigg)\OPT(\CJ),p_\mathrm{max}, (1+\delta)cB+2\frac{R(\CJ)}{1-\delta}\right).\]
Now, by lemma \Cref{le.2ndRtildeBound} and the definition of simple sequences, there holds  $\tilde R(\CJ)\le \left(1+ \frac{1}{\delta^2}\right)R(\CJ) \le (1-\delta)\frac{\delta}{2}(2-c)$. In particular, $\left(\frac{\tilde R(\CJ)}{1-\delta}+1 +2\delta\right)\OPT(\CJ)\le (c+2\delta)\OPT(\CJ)$. The second term $p_\mathrm{max}$ is always smaller than $\OPT$. Concerning the third bound in the max-term observe using \Cref{le.placeholdersize} that there holds
\begin{align*} B & = \max\left(p_\mathrm{max}^{\delta^2n}, \hat L\right) \\
&\le \max\left(p_\mathrm{max}, \delta^{-2}L\right)\\
&\le \max\left(p_\mathrm{max}, \delta^{-2}R(\CJ)p_\mathrm{max}\right)\\
&\le  \max\left(p_\mathrm{max},\delta^{-2}\frac{(1-\delta)\delta^3}{2(\delta^2+1)}(2-c)p_\mathrm{max}\right)\\&\le p_\mathrm{max}\\
&\le \OPT.\end{align*} Since $2\frac{\tilde R(\CJ)}{1-\delta}\le (2-c)\delta$ we have $(1+\delta)cB+2\frac{R(\CJ)}{1-\delta}\OPT \le (c+2\delta)\OPT$.
\end{proof}

\begin{proof}[Proof of \Cref{le.main.basic}]
\Cref{le.main.basic} follows immediately from \Cref{co.goodongeneral} and \Cref{co.goodonsimple}.
\end{proof}

\subsection{Stable job sequences. Proof sketch of  \Cref{le.main.stable}}\label{sec.ana.stable}
We introduce the class of \emph{stable} job sequences. The first two conditions state that all estimates our algorithm makes are \emph{accurate}, i.e.\ sampling works. 
By the third condition there are less huge jobs than reserve machines and the fourth condition states that these jobs are distributed evenly. The final condition is a technicality. Stable sequences are useful since they occur with high probability if we randomly order a proper job set.

\begin{definition}\label{def.stable}
A job sequence $\CJ^\sigma$ is \emph{stable} if the following conditions hold:
\begin{itemize}
\setlength{\parskip}{0pt} \setlength{\itemsep}{0pt plus 1pt}
\item The estimate $\hat L$ for $L$ is accurate, i.e.\ $(1-\delta)L \le \hat L \le (1+\delta)L$.
\item The estimate $c_p$ for $n_p$ is accurate, i.e.\ $c_p\le n_p\le c_p+2m^{3/4}$ for all $p\in \CP$.
\item There are at most $\lceil \delta m \rceil$ huge jobs in $\CJ^\sigma$.
\item Let $\tilde t$ be the time the last huge job arrived and let $n_{p,\tilde t}$ be the number of $p$-jobs scheduled at that time for a given $p\in \CP$. Then $n_{p,\tilde t}\le \left(1-\delta^3\right)n_{p}$ for every $p\in \CP$ with $n_p>\left\lfloor\frac{\left(1-\delta-2\delta^2\right)m}{|\CP|} \right\rfloor$.
\item $\delta^3 \left\lfloor \left(1-\delta-2\delta^2\right)m / |\CP| \right\rfloor \ge 2|\CP|m^{3/4}$.
\end{itemize}
\end{definition}
\begin{proof}[Proof sketch of \Cref{le.main.stable}] The first two conditions are covered by arguments following \Cref{sec.sampling}. Here, we require that only proper sequences are considered. The third condition is equivalent to demanding one of the $\lceil \delta m \rceil$ largest jobs to occur during the sampling phase. This is extremely likely. The expected rank of the largest job occurring in the sampling phase is $\delta^{-2}$, a constant. The fourth condition states that, for any $p\in\CP$, the huge jobs are evenly spread throughout the sequence when compared to any sizable class of $p$-jobs. Again, this is expected of a random sequence and corresponds to how one would view randomness statistically. For the final condition it suffices to choose the number of machines $m$ large enough.
One technical problem arises since the class $\CP=\CP[\CJ^\sigma]$ is defined using the value $B[\CJ^\sigma]$. It thus highly depends on the input permutation $\sigma$. We rectify this by passing over to a larger class $\hat\CP$ such that $\CP\subset\hat\CP$ with high probability. 
\end{proof}
The formal proof of \Cref{le.main.stable} is simple but very technical. That is, we consider the underlying ideas to be rather simple but in order to give a rigorous proof many cases have to be considered. We leave it to \Cref{sec.ana.stable.p}. The definition of stable sequences is suited for our future algorithmic arguments. To make probabilistic arguments, we introduce \emph{probabilistically stable} sequences and prove that probabilistically stable sequences are always stable. Their definition is more convenient, as it avoids certain problems such as $\CP$ being dependent on the job permutation. We then prove \Cref{le.main.stable} for all six conditions of probabilistically stable sequences separately.

\subsection{Adversarial analysis on stable sequences. Proof sketch of \Cref{le.conclusion}}
\subsubsection*{General observations}\label{sec.ana.critical.p}
This section is devoted for some general observations needed several times throughout the analysis. Recall that $\tilde L=\sup_t \frac{1}{m} \sum_M\tilde l_M^t$ denotes the maximum average load taking placeholder jobs into account. We will see that this does, in fact, not overestimate the total load $L$ if the sequence is stable.

\begin{lemma}\label{le.tildeL=L}
For every stable sequence $\CJ^\sigma$ there holds $\tilde L=L$.
\end{lemma}
\begin{proof}
By \Cref{le.sensibletildeL} we have $\tilde L\ge L$ for any pseudo-load. Recall that $\tilde L=\sup_t \frac{1}{m} \sum_M\tilde l_M^t$. Thus it suffices to show that $\sup_t \frac{1}{m} \sum_M\tilde l_M^t\le L$ for any given time $t$. Consider the schedule of our algorithm time $t$ including placeholder-jobs. If it contains $p$-placeholder-jobs for some $p\in\CP$ it contains at most $c_p$ many $p$-jobs in total. By the second property of stable sequences there holds $c_p\le n_p$. Thus, we can find real $p$-jobs not scheduled yet and replace the $p$-placeholder-jobs by them. This way the load of each machine can only increase. In particular, the resulting schedule has average load at least $\frac{1}{m} \sum_M\tilde l_M^t$. But since it contains only real jobs, its average load will be at most $L$. Therefore $\frac{1}{m} \sum_M\tilde l_M^t\le L$.
\end{proof}

The following lemma is a basic but very useful observation describing the load of any machine after the sampling phase.

\begin{lemma}\label{pro.observationalphasespec}
Let $M$ be any machine after the sampling phase and $p$ be the size of the largest job scheduled on it. Then the load of $M$ is at most $\frac{\delta^2}{1-\delta}  B+p.$
\end{lemma}
\begin{proof}
Let $l$ be the load of $M$ before the last job $J$ was scheduled on it. Using \Cref{le.avglb} we see that
$l \le \frac{m}{m-\lfloor \delta m\rfloor}\delta^2\hat L \le \frac{\delta^2}{1-\delta} B.$
The last inequality uses $L_{\delta^2}=\hat L\le  B$. Since $J$ had size at most $p$ the lemma follows.
\end{proof}

\begin{lemma}\label{co.critjobs}
Till the Least-Loaded-Strategy is employed (or till the end of the sequence) there is at most one reserve machine $M$ whose only critical job is medium. Every other machine contains either no critical job, one big job or two medium jobs (including placeholder jobs).
\end{lemma}
\begin{proof}
First consider the situation right before the Critical-Job-Strategy is employed.
Let $M$ be a machine containing a critical job. By \Cref{pro.observationalphasespec} the total size of all jobs besides the largest one on M is at most $\delta^2 \hat L\le \delta^2  B$. Since this is less than $p_\jsmall B$ only the largest job could have been critical. Now observe that the algorithm adds a second medium placeholder-job to precisely every machine that contained a (necessarily single) medium job after the sampling phase. Afterwards, medium placeholder-jobs are always scheduled in pairs onto machines which do not contain critical jobs.
While the Critical-Job-Strategy is employed, the number of medium jobs does not change for principal machines. We only replace placeholders with real jobs. Moreover the algorithm ensures that at most one reserve machine $M$ has a single medium job.
\end{proof}

Finally let us make the following technical observation, which will be necessary later. 

\begin{lemma}\label{le.memptybound}
There are at most $2\delta^2m$ machines which contain (real) critical jobs before the Preparation for the Critical-Job-Strategy. In particular $m_\mempty\le \left(1-\delta-2\delta^2\right)m$.
\end{lemma}
\begin{proof}
Assume the lemma would not hold. Since each critical job has size at least $p_\jsmall B$ this implies that $ B \ge \hat L  > \frac{1}{\delta^{2}m} \cdot 2\delta^2m \cdot p_\jsmall B=2p_\jsmall B> B$. A contradiction.
In particular, at most $2\delta^2m$ machines received critical jobs after the observational phase. Thus $m_\mempty\le |\CM|-2\delta^2m\le (1-\delta-2\delta^2)m$.
\end{proof}

\subsubsection*{Before the Least-Loaded-Strategy is employed.}\label{sec.ana.critical.p}
The goal of this section is to analyze every part of the algorithm but the Least-Loaded-Strategy. Formally we want to show the following proposition and its important \Cref{co.noll}.

\begin{restatable}{proposition}{lecmake}\label{le.cmake}
The makespan of our algorithm is at most $(c+O(\delta))\max\left( B,L,p_\mathrm{max}\right)$ on stable sequences till it employs the Least-Loaded-Strategy (or till the end of the sequence).
\end{restatable}

For a formal proof we need to consider many cases where the statement of the lemma could go wrong. Let us first give a sketch of the full proof, which will be fleshed out subsequently.

\begin{proof}[Proof sketch]
Let us only consider critical jobs at any time the Least-Loaded-Strategy, \Cref{alg.exceptional}, is not employed. Our algorithm ensures that a machine contains either one big job or at most two medium jobs. Formally, this is shown in  \Cref{co.critjobs}. In the first case, we simply bound the size of this big, possibly huge, job by $p_\jmax$. Else, if the machine contains up to two medium jobs their total size is at most $2(1+\delta)p_\jbig B = (1+\delta)c B$. The factor $(1+\delta)$ arises since we use rounded sizes in the definition of medium jobs. Thus, critical jobs may cause a load of at most $\max(p_\jmax,(c+O(\delta))B)$.

Analyzing the load increase by small, i.e.\ non-critical, jobs relies on \Cref{prop1} and~\ref{prop2} depending on whether these jobs were assigned during the sampling phase or during the Critical-Job-Strategy.
\end{proof}

Note that for stable sequences $\hat L \le (1+\delta) L \le (1+\delta)\OPT$, in particular $\max\left( B,L,p_\mathrm{max}\right)=\max\left(p^{\delta^2 n}_\jmax,\hat L, L , p_\mathrm{max}\right)\le (1+\delta)\OPT$. This proves the following important corollary to \Cref{le.cmake}.

\begin{restatable}{corollary}{collI}\label{co.noll}
Till the Least-Loaded-Strategy is used the makespan of our algorithm is at most $(c+O(\delta))\OPT$ on stable sequences.
\end{restatable}

We first need to assert that the statement holds after the preparation for the Critical-Job-Strategy, namely we prove the following proposition.

\begin{proposition}\label{le.prepan}
After the Preparation for the Critical-Job-Strategy the anticipated load of no machine exceeds $ (c+O(\delta)) B$.
\end{proposition}

There are three types of machines we need to consider. First, there are machines which contain a real critical job after the Preparation for the Critical Job Strategy. Second, there are machines, which only receive placeholder jobs.  Finally there are machines that only receive critical jobs during sampling. The following two lemmas concern themselves with the first two types of machines. Afterwards we prove \Cref{le.prepan}.

\begin{lemma}\label{le.critjob}
If a machine contains a real critical job its anticipated load is at most $( (1+\delta/2)c+\frac{\delta^2}{1-\delta}) B\le(c+O(\delta)) B$ after the Preparation for the Critical-Job-Strategy.
\end{lemma}

\begin{proof}
After the Preparation for the Critical-Job-Strategy a machine contains either a big job of size at most $B<(1+\delta/2)cB$ or two medium jobs. Each medium has size at most $(1+\delta)p_\jbig$ where the factor $(1+\delta)$ is due to rounding. Thus the total size of critical jobs is at most $2(1+\delta)p_\jbig B= (1+\delta/2)cB$. \Cref{pro.observationalphasespec} bounds the size of all non-critical jobs by $\frac{\delta^2}{1-\delta}B$.
\end{proof}

\begin{lemma}\label{le.nocritjob}
Let $M$ be the $i$-th last machine that received a job from $\CJ_\rep$ for $i\le m_\mempty$. After the Preparation for the Critical-Job-Strategy its anticipated load is at most
\[\min\left(p_\mathrm{small}+\frac{\delta^2}{1-\delta},\frac{m}{i}\delta^2\right) B +\min\left(c,\frac{m}{m_\mempty-i+1}\right) B\le(c+O(\delta)) B.\]
\end{lemma}

\begin{proof}
Let $l$ be the load of $M$ before the Preparation for the Critical-Job-Strategy and let $p$ be the sum of all the placeholder-jobs assigned to it. Then the load of $M$ after the preparation is precisely $l+p$. We bound both summands separately

To see that $l\le \min\left(p_\mathrm{small}+\frac{\delta^2}{1-\delta},\frac{m}{i}\delta^2\right) B$ observe that the largest job on $M$ has size at most $p_\mathrm{small} B$ since $M$ does contain no critical jobs. In particular, by \Cref{pro.observationalphasespec}, $l \le (p_\mathrm{small}+\frac{\delta^2}{1-\delta})B$. Consider the schedule right before placeholder jobs were assigned. By definition this schedule had average load $\frac{\delta^2}{1-\delta}\hat L \le \frac{\delta^2}{1-\delta}B$ and $M$ was at most the $i$-th most loaded machine. The second bound then follows from \Cref{le.avglb}.

We have $p\le \min\left(c,\frac{m}{m_\mempty-i+1}\right) B$. The first bound holds since we either assign two medium placeholder-jobs of size at most $p_\mathrm{big} B$ each or one big job of size at most $ B$ to any machine. Thus the sum of the placeholder-jobs assigned is at most $\max(1,2p_\mathrm{big}) B=c B$. For the second term recall that \Cref{le.placeholdersize} shows that the total size of all placeholder-jobs is at most $m \hat L\le m B$. Prior to $M$ precisely $m_\mempty-i$ machines received placeholder job of total size at least $p$. Thus, $(m_\mempty-i+1)p \le m B$, or, equivalently, $p\le \frac{m}{m_\mempty-i+1}  B$.

Altogether we derive that the anticipated load of $M$ is $l+p\le \min\left(p_\mathrm{small}+\frac{\delta^2}{1-\delta},\frac{m}{i}\delta^2\right) B +\min\left(c,\frac{m}{m_\mempty-i+1}\right) B$. We need to see that this term is in $(c+O(\delta)) B$. Consider two cases. For $i\ge \delta m$ the term is at most $\frac{m}{i}\frac{\delta^2}{1-\delta} B+c B \le (c+\delta) B$. Else, for $i\le \delta m$, it is at most $(p_\mathrm{small}+\delta^2) B+\frac{m}{m_\mempty-i+1} B\le \left(p_\mathrm{small}+\delta^2+\frac{m}{m-2\delta m-\delta^2m}\right) B=\left(c+\delta^2+\frac{2\delta-\delta^2}{1-2\delta+\delta^2}\right) B=(c+O(\delta)) B$. The first inequality uses \Cref{le.memptybound}, the second equality uses that $p_\jsmall=c-1$.
\end{proof}

\begin{proof}[Proof of \Cref{le.prepan}]
There are two cases to consider. If the machine contains a real critical job, the proposition follows from \Cref{le.critjob}. If it contains critical placeholder jobs, the proposition follows from \Cref{le.nocritjob}. Finally, if it does not receive placeholder jobs, \Cref{pro.observationalphasespec} bounds its load by $\frac{\delta^2}{1-\delta}B+p_\mathrm{max}^{\delta^2 n}\le \left(1+\frac{\delta^2}{1-\delta}\right)B$.
\end{proof}

We now come to the main result of this section.

\lecmake*

\begin{proof}
By \Cref{le.prepan} the statement of the lemma holds after the Preparation for the Critical-Job-Strategy. We have to show that it still holds afterwards. There are three cases to consider. 

First, consider reserve machines. By \Cref{le.reserve.bound} their load is at most $\max((1+\delta)c B,p_\mathrm{max})$ till the Least-Loaded-Strategy is employed.

Second, consider the case that a small job $J$ is scheduled. The job $J$ will be scheduled on a principal machine $M$ with smallest anticipated load.  By \Cref{le.avglb2} said smallest anticipated load is at most $\frac{1}{1-\delta}\tilde L$. Since $J$ has size at most $p_\jsmall B$, the anticipated load of $M$ won't exceed $\frac{1}{1-\delta}\tilde L+p_\mathrm{small} B\le (c+\frac{\delta}{1-\delta})\max\left( B,L\right)$ after $J$ is scheduled. The last inequality makes use of the fact that $\tilde L=L$ for stable sequences, \Cref{le.tildeL=L}, and that $p_\jsmall=c-1$.

Finally, we consider critical jobs that are scheduled onto principal machines. They replace placeholder-jobs. Such a critical job can have at most $(1+\delta)$-times the size of the placeholder-job it replaces. Thus it may cause the load of a machine to increase by at most $\delta  B$. Since each principal machine receives at most two critical jobs the increase on principal machines due to critical jobs is at most $2\delta B$ and the lemma follows.\footnote{A more careful analysis shows that the total increase is in fact most $c\delta B$.}
\end{proof}

\collI*

\begin{proof}Use \Cref{le.cmake} and the fact that the conditions for stable sequences imply that $B=\max\left(p^{\delta^2 n}_\jmax,\hat L\right)\le (1+\delta)\OPT$.
\end{proof}

\paragraph*{Concerning the case that the algorithm \textsc{fails}.}
We need to assert certain structural properties if the algorithm \textsc{fails}, i.e.\ reaches line 7 in \Cref{alg.main}. This is done in this section. The first important finding shows that we do not have to deal with huge jobs anymore.

\begin{restatable}{proposition}{lefailsbound}\label{le.fails.bound}
If the algorithm \textsc{fails}, every huge job has been scheduled.
\end{restatable}

The second proposition will help us obtain a lower bound on the optimum makespan.

\begin{restatable}{proposition}{lefailsboundl}\label{le.fails.bound.l}
If the algorithm \textsc{fails} at time $t$ we have $\sum_{p\in\CP} \tilde n_{p,t}w(p)>m$.
\end{restatable}


For any job class $p\in\CP$ let $c_p'$ denote the number of $p$-jobs scheduled after the Preparation for the Critical-Job-Strategy, including placeholder-jobs. This is consistent with our notation from \Cref{sec.algorithm} if we consider the values of $c_p'$ after the execution of \Cref{alg.preparation}. We call a job class $p\in\CP$  \emph{unsaturated} if $c_p'\le c_p$.
Given $p\in\CP$, let $\tilde n_{p,t}$ denote the number of $p$-jobs scheduled at any time $t$ including placeholder-jobs. After the sampling phase $\tilde n_{p,t}=\max(c_p',n_{p,t})$.

The most important technical ingredient in this chapter is to establish that if the algorithm \textsc{fails}, there is one job class of which a sizable fraction of jobs has not been scheduled even  if we take placeholder jobs into account. The next four lemmas prove this by looking at unsaturated job classes.

\begin{lemma}\label{le.saturatedjobclass}
If the algorithm \textsc{fails} on a stable sequence, there exists an unsaturated job class $p\in\CP$. In particular, during the Critical-Job-Strategy every principal machine contains either one big or two medium jobs.
\end{lemma}

\begin{proof}
Let us assume that every job class is saturated. This implies that at least $n_p-c_p$ jobs of every job class $p\in\CP$ fit onto the principal machines. By the properties of stable sequences, at most $n_p-c_p\le 2m^{3/4}$ jobs of each class thus need to be scheduled onto the reserve machines; that is at most $|\CP|\cdot 2m^{3/4}$ in total. By the last condition of stable sequences this is less than $ \lfloor\delta m\rfloor$, the number of reserve machines. A contradiction. The algorithm could not have failed.

If there was an unsaturated job class after the Preparation for the Critical-Job-Strategy, $\CJ_\rep$ must have contained precisely $m_\mempty$ elements after the second while-loop in \Cref{alg.preparation}. Else, another iteration of this loop would have added further elements. Thus, every principal machine that did not already contain real critical jobs received (critical) placeholder-jobs. By \Cref{co.critjobs} every principal machine in fact received either one big or two medium jobs.
\end{proof}

\begin{lemma}\label{le.unsatured.cpsize}
For every unsaturated job class $p\in\CP$, there holds $c_p'\ge\left\lfloor\frac{\left(1-\delta-2\delta^2\right)m}{|\CP|} \right\rfloor$.
\end{lemma}

\begin{proof}
Note that $\CJ_\rep$ actually attains cardinality $m_\mathrm{empty}$ in \Cref{alg.preparation}, otherwise there could not have been an unsaturated job class. Every time we add an element to $\CJ_\rep$ in line~4 the value~$c_p'$ increases for an unsaturated job class $p$ that currently has minimum value $c_p'$. In particular, whenever we add $|\CP|$-many elements to $\CJ_\rep$ the value $\min\limits_{p\in\CP \mathrm{ unsaturated}} c_p'$ increases by at least $1$. In total it increases at least $\left\lfloor\frac{m_\mathrm{empty}}{|\CP|}\right\rfloor$ times. The lemma follows since $m_\mathrm{empty}\le\left(1-\delta-2\delta^2\right)m$, see \Cref{le.memptybound}.
\end{proof}

\begin{lemma}\label{le.cp'sum}
There holds  $\sum_{p\in\CP} c_p'w(p)\le m-\lceil\delta m\rceil$.
\end{lemma}

\begin{proof}
 Let $n_\jmed$ be the number of medium jobs and $n_\jbig$ be the number of big jobs after the Preparation for the Critical-Job-Strategy, then $\sum_{p\in\CP} c_p'w(p)=n_\jbig + \frac{n_\jmed}{2}$. But  by \Cref{co.critjobs} every principal machine contains either one big job, two medium jobs or no critical jobs at all after the Preparation for the Critical-Job-Strategy. Reserve machines are empty. Thus, $\sum_{p\in\CP} c_p'w(p)\le |\CM| \le m-\lceil\delta m\rceil$.
\end{proof}

We now prove one main proposition of this paragraph.

\lefailsboundl*

\begin{proof}
Let $J=J_t$ be the job that caused the algorithm to fail. Consider the schedule right before job $J$ was scheduled. As a matter of thinking, let us assume that job $J$ resided on some fictional $(m+1)$-th machine $\tilde M$ at that time. We award any machine $\frac{1}{2}$ points for each medium job on it and $1$ point for each big job on it. This includes placeholder-jobs. Then $\sum_{p\in\CP} \tilde n_{p,t}w(p)$ is exactly the number of points scored by every machine including $\tilde M$. 

By \Cref{le.saturatedjobclass} every principal machine scores one point. There was no empty reserve machine, since $J$ could have been scheduled onto it, otherwise. Thus every reserve machine scores at least half a point. We call a machine \emph{bad} if it scored only $1/2$ point. There cannot be two bad reserve machines, since our algorithm would have scheduled any medium job onto such a bad machine rather than creating a second one. Moreover, if there exists a bad machine, job $J$ cannot be medium, i.e.\ $\tilde M$ cannot be not bad, too. We conclude that there is at most one bad machine amongst the $m+1$ machines, which include the fictional machine $\tilde M$. All other machines score one point. This implies that $\sum_{p\in\CP} \tilde n_{p,t}w(p)\ge m+\frac{1}{2}$.
\end{proof}

\begin{lemma}\label{le.bigjobclass}
Assume that the algorithm \textsc{fails} at time $t$ on a stable sequence and that not all huge jobs are scheduled. Then there exists a job class $p\in\CP$ with $\tilde n_{p,t}<n_p-2|\CP|m^{3/4}$
\end{lemma}

\begin{proof}
We first show that there needs to exist a job class $p\in\CP$ with $\tilde n_{p,t}> c_p'+2m^{3/4}$. Assume for contradiction sake, that we had $\tilde n_{p,t}\le c_p'+2m^{3/4}$ for every job class $p\in\CP$. Then we get a contradiction to  \Cref{le.fails.bound.l}, namely $\sum_{p\in\CP} \tilde n_{p,t}w(p)\le  \sum_{p\in\CP} c_p'w(p) + |\CP|\cdot 2m^{3/4} \le m-\lceil\delta m\rceil+ \frac{\delta^3}{2} \left\lfloor \frac{\left(1-\delta-2\delta^2\right)m}{|\CP|} \right\rfloor \le m$. The second inequality uses \Cref{le.cp'sum} and the fifth condition on stable sequences.

Thus, let $p$ be such a job class satisfying $\tilde n_{p,t}> c_p'+2m^{3/4}$. Since $\tilde n_{p,t}=\max(n_{p,t},c_p')$, this implies that $n_{p,t}> c_p'+2m^{3/4}$. Moreover, since $n_{p,t}\le c_p+2m^{3/4}$ by the second property of stable sequences, we must have $c_p'<c_p$, i.e.\ the job class $p$ is unsaturated. \Cref{le.unsatured.cpsize} implies $c_p'\ge \left\lfloor\frac{\left(1-\delta-2\delta^2\right)m}{|\CP|} \right\rfloor$. In particular $n_p\ge c_p > c_p'\ge \left\lfloor\frac{\left(1-\delta-2\delta^2\right)m}{|\CP|} \right\rfloor$. We conclude that 
$n_{p,t}\le n_p-\delta^3 n_p < n_p-\delta^3 \left\lfloor\frac{\left(1-\delta-2\delta^2\right)m}{|\CP|} \right\rfloor \le n_p-2|\CP|m^{3/4}.$ The first inequality uses the fourth condition of stable sequences, recall that by assumption not all huge jobs are scheduled; the second inequality uses the bound on $n_p$ we just derived; the final inequality uses the fifth condition of stable sequences.
\end{proof}

We finally prove \Cref{le.fails.bound}, the remaining main proposition of this paragraph.

\lefailsbound*
\begin{proof}
Let $\tilde t$ be the time the algorithm fails. By \Cref{le.bigjobclass} there exists a job class $q$ such that \[w(q)n_{q,\tilde t}<w(q)n_q-w(q)2|\CP|m^{3/4}\le w(q)n_q-|\CP|m^{3/4}\] holds. 
In particular
\begin{equation}\sum_{p\in\CP} \tilde n_{p,\tilde t}w(p) \le \sum_{p\in\CP} n_pw(p)-|\CP|m^{3/4}\le\sum_{p\in\CP} (n_p -2m^{3/4})w(p)\le \sum_{p\in\CP} c_p w(p).\end{equation}\label{eq.valuations}
The first inequality uses the previous bound on $n_{q,\tilde t}$ and the fact that for stable sequences $\tilde n_{p,\tilde t}=\max(c_p',n_{p,\tilde t})\le n_p$. For the second inequality observe that $w(p)\le 1$ for all~$p\in\CP$. For the last inequality use again the second condition of stable sequences.

Now \Cref{le.fails.bound.l} and the previous inequality imply that \[m< \sum_{p\in\CP} \tilde n_{p,\tilde t}w(p)\le \sum_{p\in\CP} c_p w(p).\] If this was the case, the algorithm would already have chosen the Least-Loaded-Strategy in \Cref{alg.preparation} line~6 and thus never failed, i.e.\ reached line~7 in \Cref{alg.main}. A contradiction.
\end{proof}

\paragraph*{The Least-Loaded-Strategy.}\label{sec.ana.leastloaded}
We now derive two important consequences from the previous section.

\begin{lemma}\label{le.ll.trivialbound}
If the input sequence is stable, the Least-Loaded-Strategy schedules every huge job onto an empty machine. Thus, if the makespan increases due to the Least-Loaded-Strategy scheduling a huge job, it is at most $p_\mathrm{max}\le\OPT$.
\end{lemma}

\begin{proof}[Proof of \Cref{le.ll.trivialbound}]
By \Cref{le.fails.bound} if a huge job is scheduled using the Least-Loaded-Strategy, our algorithm already decided to do so during the Preparation for the Critical-Job-Strategy, \Cref{alg.preparation}. At this time all $\lfloor\delta m\rfloor$ reserve machines were empty. By the conditions of stable sequences there are at most $\lfloor\delta m\rfloor$ huge jobs and there will always be an empty reserve machine available once one arrives.
\end{proof}

\begin{lemma}\label{le.ll.bound}
If our algorithm schedules a normal job $J$ using the Least-Loaded-Strategy, the load of the machine the job is scheduled on will be at most $\frac{1}{1-\delta}L+ B$. For stable sequences this is at most $\left(2+\frac{2\delta-\delta^2}{(1-\delta)^2}\right) B=(2+O(\delta)) B$.
\end{lemma}

\begin{proof}[Proof of \Cref{le.ll.bound}]
Let $l$ be the load of the machine $M$ before job $J$ was scheduled on it. Since $M$ was the least loaded principal machine at that time $l\le \frac{m}{m-\lfloor\delta m\rfloor}L\le \frac{1}{1-\delta}L$ by \Cref{le.avglb}. Since $J$ was normal, its size was at most $ B$. The first part of the lemma follows.
For the second part observe that the first condition on stable sequences implies that $L \le \frac{ B}{1-\delta}$ and thus $\frac{1}{1-\delta}L+ B\le \left(2+\frac{2\delta-\delta^2}{(1-\delta)^2}\right) B=(2+O(\delta)) B$.
\end{proof}

In order to ameliorate this worse general lower bound we need a better upper bound for $B$.

\begin{restatable}{lemma}{leOPTboundI}\label{le.OPT.bound1}
If the Least-Loaded-Strategy is applied on a stable sequence, $B \le \frac{c}{2} \OPT$.
\end{restatable}

\begin{proof}
Let us first assert that $\sum_{p\in\CP}n_pw(p)\ge \sum_{p\in\CP}\tilde n_{n,p}^tw(p)>m$.
There are two cases. If the algorithm chooses the Least-Loaded-Strategy in the Preparation for the Critical-Job-Strategy there holds  $\sum_{p\in\CP}c_p w(p)>m$. By the properties of stable sequence $c_p\le n_p$ and thus the inequality follows. Else, the algorithm \textsc{fails}, i.e.\ reaches line 7 in \Cref{alg.main}. Let $t$ be the time that happens. Then using \Cref{le.fails.bound.l} we have that
$\sum_{p\in\CP}n_pw(p)\ge \sum_{p\in\CP}\tilde n_{n,p}^tw(p)>m$.

Consider any schedule. We say a machine scores $\frac{1}{2}$ points for every medium job and $1$ point for every big job. Then $\sum_{p\in\CP}n_pw(p)>m$ is the number of points scored in total. Thus there existed a machine which scores strictly more than $1$ point. Such a machine must contain either three medium or  one big and another critical job. In the former case its load will be at least $3p_\mathrm{small} B=3(c-1) B>\frac{2}{c} B$, in the latter case its load is at least $(p_\jsmall B+p_\jbig B)=\left(\frac{3}{2}c-1\right) B=\frac{2}{c} B$. The last equality holds since $c=\frac{1+\sqrt{13}}{3}$.
\end{proof}



\paragraph*{Final proof of \Cref{le.conclusion}}\label{sec.ana.leastloaded.p}

If the algorithm does not change its makespan while applying the Least-Loaded-Strategy, the result follows form \Cref{co.noll}. If the makespan of our algorithm is caused by a huge job while applying the Least-Loaded-Strategy, it leads to an optimal makespan of $p_\mathrm{max}\le\OPT$ by \Cref{le.ll.trivialbound}. Finally, if the makespan of our algorithm is caused by a normal job it will be $(2+O(\delta)) B$ by \Cref{le.ll.bound}. On the other hand, \Cref{le.OPT.bound1} implies that $\OPT \ge \frac{2}{c} B$ in this case. The competitive ratio is thus at most
$\frac{(2+O(\delta)) B}{\frac{2}{c} B}\le c+O(\delta)$.

\section{Lower bounds}\label{sec.lowerbound}
We establish the following theorem using two lower bound sequences.
\begin{restatable}{theorem}{temainlb}\label{te.main.lb}
For every online algorithm $A$ 
 there exists a job set $\CJ$ such that

\[\bP_{\sigma\sim S_n}\left[A(\CJ^\sigma)\ge \frac{\sqrt{73}-1}{6}\OPT(\CJ)\right]\ge \frac{1}{6}.\]
This result actually holds for randomized algorithms too if the random choices of the algorithm are included in the previous probability.
\end{restatable}

\Cref{te.main.lb} implies the following lower bounds.

\begin{restatable}{corollary}{colbI}
If an online algorithm $A$ is nearly $c$-competitive,  $c\ge\frac{\sqrt{73}-1}{6}\approx 1.257$.
\end{restatable}
\vspace{-1pt}
\begin{restatable}{corollary}{colbII}
The best competitive ratio possible in the secretary~model is $\frac{\sqrt{73}+29}{36}\approx 1.043$.
\end{restatable}

Let us now prove these results. For this section let $c=\frac{\sqrt{73}-1}{6}$ be our main lower bound on the competitive ratio. We consider three types of jobs:
\begin{enumerate}
\setlength{\parskip}{0pt} \setlength{\itemsep}{0pt plus 1pt}
\item \emph{negligible jobs} of size $0$ (or a tiny size $\varepsilon>0$ if one were to insist on positive sizes).
\item \emph{big jobs} of size $1-\frac{c}{3}=\frac{17-\sqrt{37}}{18}\approx 0.581$.
\item \emph{small jobs} of size $\frac{c}{3}=\frac{1+\sqrt{37}}{18}\approx 0.419$
\end{enumerate}

Let $\CJ$ be the job set consisting of $m$ jobs of each type.

\begin{lemma}
There exists a schedule of $\CJ$ where every machine has load $1$. Every schedule that has a machine with smaller load has makespan at least $c$.
\end{lemma}

\begin{proof}
This schedule is achieved by scheduling a type~2 and a type~3 job onto each machine. The load of each machine is then $1$.
Every schedule which allocates these jobs differently must have at least one machine $M$  which contains at least three jobs of type $2$ or $3$ by the pigeonhole principle. The load of $M$ is then at least $3\frac{c}{3}=c$.
\end{proof}

Given a permutation $\CJ^\sigma$ of $\CJ$ and an online algorithm $A$, which expects $3m+1$ jobs to arrive in total. Let $A(\CJ^\sigma,3m+1)$ denote its makespan after it processes~$\CJ^\sigma$ expecting yet another job to arrive. Let $P=\bP[A(\CJ^\sigma,3m+1)=1]$ be the probability that $A$ achieves the optimal schedule where every machine has load $1$ under these circumstances. Depending on $P$ we pick one out of two input sets on which $A$ performs bad.

Let $j\in\{1,2\}$. We now consider the job set $\CJ_j$ consisting of $m$ jobs of each type plus one additional job of type $j$, i.e.\ a negligible job if $j=1$ and a big one if $j=2$. We call an ordering $\CJ_j^\sigma$ of $\CJ_j$ \emph{good} if it ends with a job of type $j$ or, equivalently, if its first $3m$ jobs are a permutation of $\CJ$. Note that the probability of $\CJ^\sigma$ being good is $\frac{m+1}{3m+1}\ge \frac{1}{3}$ for $\sigma\sim S_{3m+1}$.

\begin{lemma}
We have \[\bP_{\sigma\sim S_n}\left[A(\CJ_1^\sigma)\ge c\OPT(\CJ)\right]\ge \frac{1-P}{3}\] 
and \[\bP_{\sigma\sim S_n}\left[A(\CJ_2^\sigma)\ge c\OPT(\CJ)\right]\ge \frac{P}{3}.\] 

\end{lemma}

\begin{proof}
Consider a good permutation of $\CJ_1$. Then with probability $1-P$ the algorithm $A$ does have makespan $c$ even before the last job is scheduled. On the other hand $\OPT(\CJ_1)=1$. Thus with probability $\frac{1-P}{3}$ we have $A(\CJ_1^\sigma)=c=c\OPT(\CJ_1)$.

Now consider a good permutation of $\CJ_2$. Then, with probability $P$, algorithm $A$ has to schedule the last job on a machine of size $1$. Its makespan is thus $2-\frac{c}{3}=c^2$ by our choice of $c$. The optimum algorithm may schedule two big jobs onto one machine, incurring load $2-\frac{2c}{3}<c$, three small jobs onto another one, incurring load $c$ and one job of each type onto the remaining machines, causing load $1<c$. Thus $\OPT(\CJ_2)=c$. In particular we have with probability $\frac{P}{3}$ that $A(\CJ_2^\sigma)=c^2=c\OPT(\CJ_2)$.
\end{proof}

We now conclude the main three lower bound results.

\temainlb*
\begin{proof}
By the previous lemma we get that \[\max_{j=1,2} \left(\bP_{\sigma\sim S_n}\left[A(\CJ_j^\sigma)\ge c\OPT(\CJ)\right]\right)=\max\left(\frac{1-P}{3},\frac{P}{3}\right) \ge \frac{1}{6}.\qedhere\]
\end{proof}

\colbI*

\begin{proof}
This is immediate by the previous theorem.
\end{proof}

\colbII*

\begin{proof}
Let $A$ be any online algorithm. Pick a job set $\CJ$ according to \Cref{te.main.lb}. Then
\[
A^\rom(\CJ)=\bE_{\sigma\sim S_n}[A(\CJ^\sigma)]\ge \frac{1}{6}\cdot \frac{\sqrt{73}-1}{6}\OPT(\CJ) +\frac{5}{6}\OPT(\CJ) = \frac{\sqrt{73}+29}{36}\OPT(\CJ).\qedhere
\]
\end{proof}

\bibliographystyle{plain}
\let\oldbibliography\thebibliography
\renewcommand{\thebibliography}[1]{%
  \oldbibliography{#1}%
  \setlength{\itemsep}{0pt plus .3pt}
  \setlength{\parsep}{0pt plus .3pt}
   \setlength{\parskip}{0pt plus .3pt}
}

\bibliography{secretary_scheduling}

\appendix
\section{Missing proofs in \Cref{sec.basic}}\label{sec.basic.p}

\leavglbII*
\begin{proof}
Let $\tilde l$ be $k$-th least pseudo-load at time $t$. This means that there are at least $m-k+1$ machines with pseudo-load $\tilde l_M^t\ge \tilde l$. In particular $(m-k+1)\tilde l \le \sum_M \tilde l_M^t \le m\tilde L$.
\end{proof}

\propII*
\begin{proof}Let $\tilde l$ be the pseudo-load of the $i$-th loaded machine before $J$ is scheduled. We have $\tilde l \le \frac{m}{m-i+1}\tilde L$ by \Cref{le.avglb2}. Since $J$ had size at most $p_\mathrm{max}$, the load of the machine it was scheduled on will not exceed
$\tilde l+p_\mathrm{max} \le \frac{m}{m-i+1}\tilde L + p_\mathrm{max}\le \frac{m}{m-i+1}\frac{\tilde L}{\max(L,p_\mathrm{max})}\OPT+\OPT =  \left(\left(\frac{m}{m-i+1}\right)\tilde R(\CJ)+1 \right)OPT$.
\end{proof}

\textbf{Sampling and the Load Lemma:}

\prosample*
\begin{proof}
For $\sigma\sim S_n$ chosen uniformly randomly, the random variable $n_{\CC,\varphi}[\sigma]$ is hypergeometrically distributed: It counts how many out of $\lfloor \varphi n\rfloor$ jobs, chosen randomly from the set of all $n$ jobs without replacement, belong to $\CC$. The mean of $n_{\CC,\varphi}$ is thus \[\bE[n_{\CC,\varphi}]=\frac{\lfloor \varphi n \rfloor}{n} n_\CC ,\]
in particular, using that $\CJ$ has size at least $n\ge m$, we have
\begin{equation}0\le \varphi n_\CC -\bE[n_{\CC,\varphi}] \le \frac{1}{n} \le \frac{1}{m}.\end{equation} Similarly, the variance of $n_{\CC,\varphi}$ is at most \begin{align*}\textrm{Var}[n_{\CC,\varphi}] &=\frac{n_\CC\left(n-n_\CC\right)\lfloor\varphi n \rfloor \left(n-\lfloor\varphi n \rfloor\right) }{n^2(n-1)}\le \varphi n_\CC
.\end{align*} By Chebyshev's inequality we have:
\begin{align*}&\bP_{\sigma\sim S_n}\left[\left|\varphi^{-1} n_{\CC,\varphi}\left[\sigma\right]-n_{\CC}[\CJ]\right|\ge E\right]\\
\le &\bP_{\sigma\sim S_n}\left[\left| n_{\CC,\varphi}\left[\sigma\right]-\bE\left[n_{\CC,\varphi}\right]\right|\ge \varphi\left(E-1/m\right)\right]\\
\le &\frac{\textrm{Var}[n_{\CC,\varphi}]}{\varphi^2(E-1/m)^2}\\
\le &\frac{n_\CC}{\varphi (E-1/m)^2}&&\qedhere \end{align*}
\end{proof}

\Loadlemma*
\begin{proof}
Let $\delta=\frac{\varepsilon}{2}$ and let $F(m)=\frac{\sqrt{mR_\mathrm{low}}\delta}{1+\delta}$. The assumption that $\varepsilon^{-4}\varphi^{-1}R_\mathrm{low}^{-1}=o(m)$ already implies $F\in\Theta\left(\sqrt{mR_\mathrm{low}}\varepsilon\right)\subset \omega\left(\frac{1}{\varepsilon\sqrt{\varphi}}\right)$.

Let  us fix any input sequence $\CJ$.
Given a non-negative integer $j\in\BZ_{\ge 0}$ let $p_j=(1+\delta)^{-j}p_\mathrm{max}$. For $j>0$ let $\CC_j$ denote the set of jobs in $\CJ$ that have size in the half-open interval $(p_{j},p_{j-1}]$. Note that every job belongs to precisely one job class $\CC_j$. Using the notation from \Cref{sec.sampling} we set $n_{j}^\sigma=n_{\CC_j,\varphi}[\sigma]$ and $n_j=n_{\CC_j}$.

Now $L^\downarrow= \frac{1}{m}\sum\limits_{j=1}^\infty (1+\delta)^{-j}n_jp_\mathrm{max}$ is the average load if we round down the size of every job in job class $\CC_j$ to $p_j$ for every $j\ge 1$. In particular, there holds $L^\downarrow\le L \le (1+\delta)L^\downarrow$.  Similarly, let $L^\downarrow_\varphi[\sigma]= \frac{\varphi^{-1}}{m}\sum\limits_{j=1}^\infty (1+\delta)^{-j}n_j^\sigma p_\mathrm{max}$ be the rounded-down version of $L_\varphi$. Again, there holds $L^\downarrow_\varphi[\sigma]\le L_\varphi[\sigma] \le (1+\delta)L^\downarrow_\varphi[\sigma]$. Using these approximations, we see that
\begin{align*}\left|L_\varphi[\sigma]-L\right|\le \left|L^\downarrow_\varphi[\sigma]-L^\downarrow\right|+\max\left(\delta L^\downarrow_\varphi[\sigma],\delta L^\downarrow\right) \\
\le \left|L^\downarrow_\varphi[\sigma]-L^\downarrow\right|+\delta \left|L^\downarrow_\varphi[\sigma]-L^\downarrow\right| +\delta L^\downarrow \\
\le (1+\delta)\left|L^\downarrow_\varphi[\sigma]-L^\downarrow\right| +\delta L\end{align*}

We can bound the term in the statement of the lemma via
\begin{equation}\label{eq.Lcompare}\left|\frac{L_\varphi}{L}-1\right|\le (1+\delta)\frac{\left|L^\downarrow_\varphi[\sigma]-L^\downarrow\right|}{L}+\delta.
\end{equation}

Now, consider the term $\left|L^\downarrow_\varphi[\sigma]-L^\downarrow\right|$. \Cref{pro.sample} with $E=(1+\delta)^{j/2}F(m)\sqrt{n_j}$, yields
\[\bP_{\sigma\sim S_n}\left[\left|\varphi^{-1}n_j^\sigma-n_j\right|\ge (1+\delta)^{j/2}F(m)\sqrt{n_j}\right]=O\left(\frac{(1+\delta)^{-j}}{\varphi F(m)^2}\right)\]
Consider the event that we have $\left|\varphi^{-1}n_j^\sigma-n_j\right|\ge (1+\delta)^{j/2}F(m)\sqrt{n_j}$ for all $j$. By the union bound its probability is 
\[P(m)=1-O\left(\sum\limits_j\frac{(1+\delta)^{-j}}{\varphi F(m)^2}\right)= 1-O\left(\frac{1}{\delta\varphi F(m)^2}\right)=1-o(\varepsilon).\]
The first equality uses the union bound, the second the geometric sequence and the final one the fact, argued at the beginning of the proof, that $F\in \omega\left(\frac{1}{\varepsilon\sqrt{\varphi}}\right)$ and that $\delta=\Theta(\varepsilon)$.
Now, if we have $\left|\varphi^{-1}n_j^\sigma-n_j\right|\ge (1+\delta)^{j/2}F(m)\sqrt{n_j}$ for all $j\ge 1$, we get:
\begin{align*}\left|L^\downarrow_\varphi[\sigma]-L^\downarrow\right|&=\left|\frac{\varphi^{-1}}{m}\sum\limits_{j=0}^\infty (1+\delta)^{-j}n_j^\sigma p_\mathrm{max}-\frac{1}{m}\sum\limits_{j=0}^\infty (1+\delta)^{-j}n_jp_\mathrm{max}\right| \\
&\le \frac{1}{m}\sum\limits_{j=0}^\infty (1+\delta)^{-j}\left|\varphi^{-1}n_j^\sigma-n_j\right| p_\mathrm{max}\\
&\le \frac{\sqrt{p_\mathrm{max}}}{m}\sum\limits_{j=0}^\infty (1+\delta)^{-j/2}F(m)\sqrt{n_j}\sqrt{p_\mathrm{max}}\\
&\le\frac{\sqrt{p_\mathrm{max}}F(m)}{\sqrt{m}}\left(\frac{1}{m}\sum\limits_{j=0}^\infty (1+\delta)^{-j}n_jp_\mathrm{max}\right)^{1/2}\\
&\le\frac{\sqrt{L}F(m) }{\sqrt{R_\mathrm{low} m}}\sqrt{L^\downarrow}\\
&\le\frac{F(m)}{\sqrt{R_\mathrm{low} m}}L =\frac{\delta}{1+\delta}L.
\end{align*}
The first inequality is the triangle inequality, the second holds per assumption, the third is the Cauchy–Schwarz inequality, the fourth uses the definition of $L^\downarrow$ and the fact that $\sqrt{p_\mathrm{max}}=\sqrt{{L}/{R[\CJ]}}\le \sqrt{{L}/{R_\mathrm{low}}}$, the last inequality uses that $L^\downarrow\le L$ and the final equality is simply the definition of $F$.

Combining this with \Cref{eq.Lcompare} yields that we have with probability $P(m)\in 1-o(\varepsilon)$
\[\left|\frac{L_\varphi}{L}-1\right|\le (1+\delta)\frac{\left|L^\downarrow_\varphi[\sigma]-L^\downarrow\right|}{L}+\delta\le 2 \delta =\varepsilon.\]
It thus suffices to choose $m(R_\mathrm{low},\varphi,\varepsilon)$ such that $P(m)\le 1-\varepsilon$ for all $m\ge m(R_\mathrm{low},\varphi,\varepsilon)$.
\end{proof}
\section{Missing proofs in \Cref{sec.1.75}}\label{sec.1.75.p}
\coll*
\begin{proof}[Proof of \Cref{co.ll}]
Let wlog.\ $\CJ^\sigma=J_1,\ldots, J_n$ and set $\OPT=\OPT(\CJ^\sigma)=\OPT(J_1,\ldots, J_n)$. We append a certain number of jobs $J_{n+1},\ldots, J_{n'}$ to the sequence such that the average load of $J_{1},\ldots, J_n'$ is $L_\mathrm{guess}$ and $\OPT(J_{1},\ldots, J_n')=\max(L_\guess,\OPT)$. We use the following procedure to construct the sequence:
\begin{algorithm}[H]
\caption{Appending a certain job sequence.}\label{alg.exceptional3}
\begin{algorithmic}[1]
\State Start with $n'=n$ and any optimal schedule of $J_{1},\ldots, J_n$.
\While{$L_{n'}=\frac{1}{m}\sum\limits_{i=1}^{n'} p_{n'} <L_\guess$}
\State Let $M$ be a least loaded machine and $l$ be its load.
\State Append job $J_{n'+1}$ of size $p_n'=\min(L_\guess - l, m(L_\guess - L_{n'}))$ to the sequence.
\State Schedule $J_{n'}$ onto $M$. $n'\gets n'+1$.
\EndWhile
\end{algorithmic}
\end{algorithm}

It is easy to see that the previous schedule outputs a job sequence of average load $L[J_1,\ldots, J_{n'}]=L_{n'}=L_\guess$, assuming it started with a sequence $J_1,\ldots, J_n$ of average load at most $L_\guess$. Furthermore it maintains a schedule with makespan at most $\max (L_\guess,\OPT)$. This is necessarily an optimal schedule, since both the average load $L_\guess$ as well as the optimum $\OPT=\OPT(J_1,\ldots, J_n)$ of a prefix are lower bounds on the optimum makespan. We've thus shown that $\OPT(J_{1},\ldots, J_{n'})=\max(L_\guess,\OPT)$.

Since we can apply \Cref{te.ll} to $J_{1},\ldots, J_{n'}$ we get:
\begin{align*}\mathrm{Light Load [L_\guess]}(J_1,\ldots, J_n)&\le \mathrm{Light Load [L[J_1,\ldots, J_{n'}]]}(J_{1},\ldots, J_{n'}) \\ &\le 1.75 \OPT(J_{1},\ldots, J_{n'})\\  &= 1.75\max(L_\guess,\OPT). &\qedhere\end{align*}
\end{proof}

\section{Second reduction. Full proof of \Cref{le.main.stable}}\label{sec.ana.stable.p}
Throughout this proof we assume that all job sets $\CJ$ considered are proper. Many notations in this proof will depend on the job set $\CJ$, the number of machines $m$ and possibly the job order. For simplicity we omit these dependencies whenever possible. If needed, we include it using the notation $\CP[\CJ^\sigma]$ or even $\CP[\CJ^\sigma,m]$, $\hat n_p[\CJ^\sigma]$, etc. In particular, we write mostly $\delta$ for the function $\delta(m)=\frac{1}{\log(m)}$. It is very important to note for the arguments in this section, that this is a function in $m$ whose inverse grows sub-polynomially in $m$.

For every job set $\CJ$ we fix a set $S=S[\CJ,m]\subset\CJ$ consisting of the $\left\lceil\delta(m)^{-7/3}\right\rceil$ largest jobs. We solve ties arbitrarily. Technically, we could choose any exponent other than $7/3$ in the open interval $(2,3)$, too. Let $s_\mathrm{min}=s_\mathrm{min}[\CJ,m]$ be the size of the smallest job in the set $S$. Recall the geometric rounding function $f(p)=(1+\delta)^{\left\lfloor \log_{1+\delta}p\right\rfloor}$ and consider the set $\CP_\mathrm{glob}=\CP_\mathrm{glob}[\CJ,m]=\{f(p_t)\mid p_t\text{ is the size of any job }J_t\in\CJ\}$ of all rounded sizes of jobs in $\CJ$. Then consider the subset $\hat\CP=\hat\CP[\CJ,m]=\{(1+\delta)^i\in\CP_\mathrm{glob} \mid p_\jsmall \max(s_\mathrm{min},(1-\delta)L)< (1+\delta)^i \}$. We will see that this set is likely a superset of $\CP[\CJ^\sigma,m]$, which does not depend on the job order. 

The following estimates of the sizes of $\CP$ and $\hat \CP$ will be relevant later.

\begin{lemma}\label{le.Pbound}
We have $|\CP|\le 1-\lfloor\log_{1+\delta}(p_\jsmall)\rfloor \le O\left(\delta^{-1}\right)$. 
\end{lemma}
\begin{proof}
First observe that $\CP$ contains precisely one element for each power of $(1+\delta)$ in the half-open interval $((1+\delta)^{-1}p_\jsmall B, B]$. In particular $|\CP|\le 1-\lfloor\log_{1+\delta}(p_\jsmall)\rfloor \le O\left(\delta^{-1}\right)$. 
\end{proof}

\begin{lemma}\label{le.hatCPbound}
We have $|\hat \CP|\le \delta(m)^{-7/3}-\lfloor\log_{1+\delta}(p_\jsmall)\rfloor\le O(\delta^{-7/3})$
\end{lemma}

\begin{proof}
Indeed, there are precisely $1-\lfloor\log_{1+\delta}(p_\jsmall)\rfloor$ powers of $(1+\delta)$ in the half-open interval $(p_\jsmall s_\mathrm{min},s_\mathrm{min}]$, in particular $\hat\CP$ contains at most that many elements of size lesser or equal to $s_\mathrm{min}$. Now all elements in $\hat\CP$ which have size strictly greater than $s_\mathrm{min}$ need to be the rounded sizes of jobs in $S$ excluding the smallest job in $S$. Thus, there are at most $\delta(m)^{-7/3}-1$ elements in $\hat\CP$ of size strictly greater than $s_\mathrm{min}$. In particular $|\hat \CP|\le 1-\lfloor\log_{1+\delta}(p_\jsmall)\rfloor+ \delta(m)^{-7/3}-1$.
\end{proof}

For every $p\in\hat\CP[\CJ,m]$ we consider the job class $\CC_p=\CC_p[\CJ,m]\subset\CJ$ of jobs whose rounded size is $p$. Using the notation from \Cref{sec.sampling} we set
$n_{p,\varphi}=n_{\CC_p,\varphi}$ and $n_p=n_{\CC_p}$ for every $p\in\hat\CP$ and $0<\varphi<1$. We defined the property of being stable in a way that lends itself to algorithmic applications. We now give similar, in fact slightly stronger, conditions better suited for a probabilistic arguments.

\begin{definition}
We call a proper job sequence $\CJ^\sigma$ \emph{probabilistically stable} if the following holds:
\begin{enumerate}
\setlength{\parskip}{0pt} \setlength{\itemsep}{0pt plus 1pt}
\item The load estimate $\hat L=\hat L_{\delta^2}[\CJ^\sigma]$ for $L=L[\CJ]$ is good, i.e.\ $(1-\delta)L \le \hat L \le (1+\delta)L$.
\item There is at least one job in $S$ among the $\lceil\delta^2 n\rceil$ first jobs in $\CJ^\sigma$.
\item For every $p\in\hat\CP$ we have $\left|\delta^{-2} n_{p,\delta^2}-n_p\right|\le m^{3/4}-1$.
\item For every $p\in\CP_\jmed[\CJ^\sigma]$ we have $2n_{p,\delta^2}\le n_p$.
\item Let $t_S=t_S[\CJ^\sigma,m]$ be the time the last job in $S$ arrived, then $t_S \le  \left(1-\delta(m)^{8/3}\right)n$.
\item For every $p\in\hat \CP$ with $n_p>\left\lfloor\frac{\left(1-\delta-2\delta^2\right)m}{|\CP|} \right\rfloor$ there holds $n_{p,1-\delta^{8/3}}\le \left(1-\delta^3\right)n_{p}$.
\end{enumerate}
We refer to these six conditions as \emph{probabilistic conditions}.
\end{definition}

The following lemma is shows that we can analyze the probability of a sequence being probabilistically stable instead of using the conditions from \Cref{sec.ana.stable}.

\begin{lemma}\label{le.stableprobal}
There exists a number $m_0$ such that for all $m\ge m_0$ every probabilistically stable sequence is stable.
\end{lemma}

The constant $m_0$ in the previous lemma comes from the following technical lemma.

\begin{lemma}\label{le.m0.bound}
There exists $m_0>0$, such that for all $m\ge m_0$ and all proper job sets $\CJ$ we have $\delta(m)^{-7/3}\le \lfloor\delta m\rfloor$ and $\delta(m)^3 \cdot \left\lfloor \frac{\left(1-\delta(m)-2\delta(m)^2\right)m}{|\CP(m,\CJ)|} \right\rfloor \ge 2|\CP(m,\CJ)|m^{3/4}$.
\end{lemma}

\begin{proof}
This comes down to asymptotic observations. For the first inequality use that $\delta(m)^{-7/3}=\log^{7/3}(m)=o(m)$ while $\lfloor\delta m\rfloor=\Theta(m)$.

For the second inequality observe that $|\CP(m,\CJ)|=O(\delta(m)^{-1}) =O(\log(m))$ by \Cref{le.Pbound}. Then we can see that
$\delta(m)^3 \cdot \left\lfloor \frac{\left(1-\delta(m)-2\delta(m)^2\right)m}{|\CP(m,\CJ)|} \right\rfloor=\Omega (\delta(m)^4m)=\Omega\left(\frac{m}{\log^4(m)}\right)$ while, on the other hand, $2|\CP|m^{3/4}=O\left(\log(m)m^{3/4}\right)=o\left(\frac{m}{\log^4(m)}\right)$.

These asymptotic observations already imply the statement of the lemma.
\end{proof}

\begin{proof}[Proof of \Cref{le.stableprobal}]
We consider the five conditions of stable sequences separately

\textbf{1.} The first condition of stable sequences agrees with the first probabilistic condition.

\textbf{2.} First consider job classes $p\in\CP[\CJ^\sigma,m]\setminus\CP_\mathrm{glob}[\CJ,m]$. For these job classes there holds $n_p=0$. This already implies that $c_p=0$ holds, too. Thus, the second condition follows trivially for these job classes.

By the second probabilistic condition, we have that $p^{\varphi n}_\mathrm{max}[\CJ^\sigma]\ge s_\mathrm{min}$ and by the first probabilistic condition there holds $\hat L[\CJ^\sigma]\ge(1-\delta)L$. In particular $B[\CJ^\sigma] =\max\left(p^{\varphi n}_\mathrm{max},\hat L\right) \ge \max(s_\mathrm{min},(1-\delta)L)$ and thus $\CP(\CJ^\sigma,m)\cap \CP_\mathrm{glob}\subseteq \hat \CP(\CJ,m)$. There are two cases to consider now. If $c_p=\left\lfloor\left(\delta^{-2}\hat n_{p}-m^{3/4}\right)w(p)\right\rfloor w(p)^{-1}$ we conclude, using the third probabilistic condition, that $|(c_p+m^{3/4})-n_p| \le |\delta^{-2}\hat n_{p}-n_p+1|\le m^{3/4}$, which already implies that $c_p\le n_p \le c_p+2m^{3/4}$. If $c_p=\hat n_{p} w(p)^{-1}\ge \left\lfloor\left(\delta^{-2}\hat n_{p}-m^{3/4}\right)w(p)\right\rfloor$ the second bound  $n_p\le c_p+2m^{3/4}$ still holds. The first bound, $c_p\hat n_{p} w(p)^{-1}\le n_p$ is trivial if $w(p)^{-1}=1$, or, equivalently, if $p\notin\CP_\jmed[\CJ^\sigma]$. Else, it follows from the fourth probabilistic condition.

\textbf{3.}  To conclude the third condition of stable sequences note that the second probabilistic condition implies that all huge jobs have size strictly greater than $s_\mathrm{min}$. This implies that they lie in $S$ and that there are at most $|S|=\delta(m)^{-7/3}$ many of those jobs. Since we only consider $m\ge m_0$ we have $\delta(m)^{-7/3}\le \lfloor \delta m\rfloor$ by \Cref{le.m0.bound}. Hence, the third condition of stable sequences follows.

\textbf{4.}  Consider $p\in\CP(\CJ^\sigma,m)$ with $n_p>\left\lfloor\frac{\left(1-\delta-2\delta^2\right)m}{|\CP|} \right\rfloor$. As we argued when proving the second condition $p\in \hat \CP(\CJ,m)$. By the second probabilistic condition all huge jobs lie in $S$ since they have size strictly greater than $\hat B \ge p^{\varphi n}_\mathrm{max} \ge s_\mathrm{min}$. It thus suffices to show that at most $(1-\delta^3)n_p$ jobs arrived at time $t_S[\CJ^\sigma,m]$, the time the last job in $S$ arrived. But by the fifth probabilistic condition this value is at most $n_{p,1-\delta^{8/3}}$, which is less than $(1-\delta^3)n_p$ by the sixth probabilistic condition. The fourth condition for stable sequences follows.

\textbf{5.} Finally, the fifth condition of stable sequences is already a consequence of choosing $m\ge m_0$ and \Cref{le.m0.bound}.
\end{proof}

Now, we analyze the probability of each probabilistic condition separately. Namely, we consider
\[P_i(m)=\sup\limits_{\CJ \textrm{ proper}}\bP_{\sigma\sim S_n}\left[\CJ^\sigma \textrm{does not fulfill the }i\text{-th condition}\right].\]
Recall, that $P(m)$ similarly defines the worst probability with which a sequence may not be stable. It is the value we are interested in. The values $P_i(m)$ relate to $P(m)$ by the following corollary, which is an immediate consequence of \Cref{le.stableprobal} and the union bound.

\begin{corollary}\label{le.p1}
We have that $P(m)\le \sum_{i=1}^6P_i(m)$ for all $m\ge m_0$ if we choose $m_0$ as in \Cref{le.m0.bound}.
\end{corollary}

Thus, we are left to see that all the $P_i(m)$ vanish.

\begin{lemma}\label{le.p2}
For every $i$ we have $\lim\limits_{m\rightarrow\infty} P_i(m)=0$.
\end{lemma}

\begin{proof}
We again consider every choice of $1\le i \le 6$ separately.

\textbf{1.} Apply the Load Lemma, \Cref{Loadlemma}, with $R_\mathrm{low}=\frac{(1-\delta)\delta^3}{2(\delta^2+1)}(2-c)$, $\varphi=\delta^2$ and $\varepsilon=\delta$. Then for $m$ large enough, there holds $ \bP_{\sigma\sim S_n} \left[\left|\frac{L_{\delta^2}[\CJ^\sigma]}{L[\CJ]}-1\right|\ge\delta\right]<\delta$. Note that the condition $\left|\frac{L_{\delta^2}[\CJ^\sigma]}{L[\CJ]}-1\right|\ge\delta$ is equivalent to $(1-\delta) L \le \hat L \le (1+\delta)L$. Thus $P_1(m)\le\delta(m)$ and in particular $\lim\limits_{m\rightarrow\infty}P_1(m)=0$.

\textbf{2.} Let $J\in S$. The probability that $J$ is not among the $\lfloor \delta^2 m \rfloor$ first jobs is at most $1-\delta^2$ after random permutation. In particular the probability $P_2(m)$ that none of the jobs in $S$ is among the first $\lfloor \delta^2 m \rfloor$ jobs can be bounded via $(1-\delta^2)^{|S|} \le \frac{1}{1+\delta^2|S|}$ using Bernoulli's inequality. Thus $P_2(m)\le \frac{1}{1+\delta(m)^2|S|} \le  \frac{1}{1+\delta(m)^{-1/3}} \le \delta(m)^{1/3}$ which tends to $0$ for $m\rightarrow\infty$.

\textbf{3.} Fix $p\in\hat \CP$. By \Cref{pro.sample} we have $\bP_{\sigma\sim S_n}\left[\left|\delta^{-2} n_{p,\delta^2}-n_p\right|\ge m^{3/4}-1\right]\le\frac{n_p}{\delta^2(m^{3/4}-1-1/m)^2}$ and thus by the union bound $P_3(m)\le \frac{\sum_{p\in\hat\CP} n_p}{\delta^2(m^{3/4}-1-1/m)^2}$. Now observe that there holds $L\ge \frac{1}{m}\sum_{p\in\hat\CP} p\cdot n_p \ge \frac{1}{m}\sum_{p\in\hat\CP}p_\jsmall (1-\delta) L \cdot n_p$, thus $\sum_{p\in\hat\CP} n_p \le \frac{m}{(1-\delta)p_\jsmall}$. From this we conclude that $P_3(m) \le \frac{m}{(1-\delta)p_\jsmall\delta^2(m^{3/4}-1-1/m)^2}=O\left(\frac{1}{\sqrt{m}}\right)$ which already shows that $\lim\limits_{m\rightarrow\infty} P_3(m)=0$.

\textbf{4.} 
Recall that $f$ is the geometric rounding function. Let $\hat B[\CJ^\sigma]=\max((1+\delta)L,f(B[\CJ^\sigma]))$ and let $\hat\CP_\jmed[\CJ^\sigma]=\{p\in\CP_\mathrm{glod}\mid (1+\delta)^{-3}p_\jsmall \hat B[\CJ^\sigma] \le p \le p_\jbig \hat B[\CJ^\sigma]\}$. We have $\CP_\jmed[\CJ^\sigma]\subset \hat\CP_\jmed[\CJ^\sigma]$ if the first probabilistic condition holds, i.e.\ $(1-\delta)L\le \hat L[\CJ^\sigma] \le (1+\delta)L.$ Since the probability of the latter not being true is $P_1(m)$ and vanishes for $m\rightarrow\infty$ it suffices to consider the second probabilistic condition where we replace $\CP_\jmed[\CJ^\sigma]$ with $\hat\CP_\jmed[\CJ^\sigma]$.

Now let $\hat B_\mathrm{fix}$ be any possible value the variable $\hat B[\CJ^\sigma]$ may obtain. We want to condition ourselves on the case that  $\hat B[\CJ^\sigma]=\hat B_\mathrm{fix}$ for $\hat B_\mathrm{fix}$ either $(1+\delta)L$ or the rounded size of some element in $S$ exceeding $(1+\delta)L$. Note, that $\hat B[\CJ^\sigma]$ may attain other values than the previously mentioned ones only if either the first or second probabilistic condition is not met. Since $P_1(m)+P_2(m)\rightarrow 0$ we may ignore these cases. Fixing $\hat B$ also fixes the set $\hat\CP_\jmed$.

Consider $p\in\hat\CP_\jmed$. We want an upper bound on $\bE[n_{p,\delta^2}\mid\hat B=\hat B_\mathrm{fix}]$, the expected value of $n_{p,\delta^2}$ conditioned on our choice of $\hat B$. What does it mean to condition on $\hat B[\CJ^\sigma]=B_\mathrm{fix}$? It is simply equivalent to stating that no element in $S$ of rounded size strictly greater than $B_\mathrm{fix}$ occurs in the sampling phase and that either an element of rounded size $B_\mathrm{fix}$ occurs in the sampling phase or $B_\mathrm{fix}=(1+\delta)L$. Thus conditioning on $\hat B[\CJ^\sigma]=B_\mathrm{fix}$ just fixes the position of some of the $\lceil\delta(m)^{-7/3}\rceil$ elements in $S$. The expected value of $n_{p,\delta^2}$ is maximized if we consider the case where all elements of $S$ have to occur after the sampling phase. In this case, $n_{p,\delta^2}$ is hypergeometrically distributed. We sample $\delta^2 n$ jobs from a set of $n-|S|$ jobs and count the number of $p$ jobs. Thus, we see that $\bE[n_{p,\delta^2}\mid\hat B=\hat B_\mathrm{fix}]\le \frac{\lceil\delta^2 n\rceil n_p}{n-\lceil\delta(m)^{-7/3}\rceil}\le 2\delta^2 n_p$. For the latter inequality we need to choose $m$ (and thus $n\ge m$) sufficiently large. Now, we get by Markov's inequality
\[\bP\left[n_{p,\delta^2}\ge\frac{n_p}{2}\right]\le \bP\left[n_{p,\delta^2}\ge\frac{\bE[n_{p,\delta^2}]}{4\delta^2}\right]\le 4\delta^2.\]
We have $|\hat\CP_\jmed[\CJ^\sigma]|\le \log_{1+\delta}(p_\jbig)-\log_{1+\delta}(p_\jsmall)+4 =O(\delta(m))$. Thus, by the union bound $P_4(m)\le |\hat\CP_\jmed|\cdot 4\delta^2 =O((\delta(m))^{-1})$ if we condition ourselves on $\hat B[\CJ^\sigma]=\hat B_\mathrm{fix}$. Since this holds for all possible choices of $\hat B_\mathrm{fix}$ but few degenerate ones that occur with probability at most $P_1(m)+P_2(m)$ we get that $P_4(m)\le P_1(m)+P_2(m)+O(\delta(m))$. Since all these terms vanish for $m\rightarrow\infty$ we also have that $\lim\limits_{m\rightarrow\infty} P_4(m)=0$.

\textbf{5.} Let $l=\left\lceil\delta(m)^{-7/3}\right\rceil$, then there holds $P_5(m)=1- \prod_{i=1}^l \frac{\left\lfloor\left(1-\delta(m)^{8/3}\right)n\right\rfloor-i}{n-i}$. Indeed, if we choose any order $S=\{s_1,\ldots,s_l\}$ then the $i$-th term in the product denotes the probability that the $i$-th element $s_i$ is among the $\lceil(1-\delta(m)^{8/3})n\rceil$-last elements conditioned on the earlier elements already fulfilling this condition. But now we see using Bernoulli's inequality that $\prod_{i=1}^l \frac{\left\lfloor\left(1-\delta(m)^{8/3}\right)n\right\rfloor-i}{n-i} \ge \left(1-\delta(m)^{8/3}-\frac{l+1}{n}\right)^l \ge 1 -\delta(m)^{8/3}l-\frac{l(l+1)}{n}$. Since for proper job sets $n\ge m$, this implies $P_5(m)\le \delta(m)^{8/3}l-\frac{l(l+1)}{m}=O(\delta(m)^{1/3})$. In particular $\lim\limits_{m\rightarrow\infty} P_5(m)=0$.

\textbf{6.} Fix any $p\in\hat\CP$ with $n_p>\left\lfloor\frac{\left(1-\delta-2\delta^2\right)m}{|\CP|} \right\rfloor=\Omega(m)$. Then we have by \Cref{pro.sample} \begin{align*}\bP\left[n_{p,1-\delta^{8/3}}\ge \left(1-\delta^3\right)n_{p}\right]&\le \bP\left[\left|(1-\delta^{8/3})^{-1}n_{p,1-\delta^{8/3}}-n_p\right|\ge \frac{\delta^{8/3}(1-\delta^{1/3})n_p}{1-\delta^{8/3}}\right]\\
&\le \frac{n_p}{(1-\delta^{8/3})\left(\frac{\delta^{8/3}(1-\delta^{1/3})n_p}{1-\delta^{8/3}}-1/m\right)^2}\\
&=O\left(\frac{\delta^{-16/3}}{n_p}\right)
=O\left(\frac{\delta^{-16/3}}{m}\right).
\end{align*}
By the union bound and \Cref{le.hatCPbound} there holds that $P_6(m)=O(|\hat \CP| \delta^{-16/3}/ m)=O(\delta^{-23/3}/m)$. Thus, $\lim\limits_{m\rightarrow\infty} P_6(m)=0$.
\end{proof}

\begin{proof}[Proof of \Cref{le.main.stable}]
By \Cref{le.p1} we have $P(m)\le\sum_{i=1}^6 P_i(m)$ for $m\ge m_0$ and by ~\Cref{le.p2} we have that $\lim\limits_{m\rightarrow\infty}\sum_{i=1}^6 P_i(m)=0$. Thus $\lim\limits_{m\rightarrow\infty} P(m)=0$.
\end{proof}

\end{document}